\documentclass[11pt]{amsart}

\usepackage[english]{babel}
\usepackage{graphicx,amsmath,amssymb,color, enumerate}
\numberwithin{equation}{section}
\usepackage{enumerate}

 \usepackage{xcolor}
 \definecolor{brickred}{rgb}{0.8, 0.25, 0.33}
\definecolor{blue(ryb)}{rgb}{0.01, 0.28, 1.0}
\definecolor{brandeisblue}{rgb}{0.0, 0.44, 1.0}
\definecolor{ceruleanblue}{rgb}{0.16, 0.32, 0.75}
\definecolor{cobalt}{rgb}{0.0, 0.28, 0.67}
\definecolor{coolblack}{rgb}{0.0, 0.18, 0.39}
\definecolor{darkblue}{rgb}{0.0, 0.0, 0.55}
\usepackage[hidelinks]{hyperref}
\hypersetup{
    colorlinks,
    citecolor=darkblue,
    filecolor=black,
    linkcolor=darkblue,
    urlcolor=black
}

\newtheorem{theorem}{Theorem}[section]
\newtheorem{lemma}[theorem]{Lemma}
\newtheorem{proposition}[theorem]{Proposition}
\newtheorem{corollary}[theorem]{Corollary}

\theoremstyle{remark}

\newtheorem{remark}[theorem]{Remark}
\newtheorem{assumption}[theorem]{Assumption}

\newtheorem{example}{\it Example\/}

\renewcommand{\Re}{\mathsf{Re}}

\newcommand{\supp}{\mathsf{supp}\,}
\newcommand{\spe}{\mathsf{sp}_{\mathsf{ess}}}

\newcommand{\N}{\mathbb{N}}
\newcommand{\Z}{\mathbb{Z}}
\newcommand{\R}{\mathbb{R}}
\newcommand{\C}{\mathbb{C}}

\newcommand{\cB}{\mathcal B}

\DeclareMathOperator{\curl}{curl}

\newcommand{\half}{\frac{1}{2}}
\title[On the magnetic Dirichlet-to-Neumann eigenvalues]{Asymptotics for the magnetic  Dirichlet-to-Neumann eigenvalues in general domains}

\author{Bernard Helffer}
\address[B. Helffer]{Laboratoire de Math\'ematiques Jean Leray, CNRS, Nantes Universit\'e,
	44000 Nantes, France.}
\email{Bernard.Helffer@univ-nantes.fr}

\author{Ayman Kachmar}
\address[A. Kachmar]{School of Science and Engineering, The Chinese University of Hong Kong Shenzhen, Guangdong, 518172, P.R. China.}
\email{akachmar@cuhk.edu.cn}

\author{Fran\c cois Nicoleau}
\address[F. Nicoleau]{Laboratoire de Math\'ematiques Jean Leray, CNRS, Nantes Universit\'e,
	44000 Nantes, France.}
\email{francois.nicoleau@univ-nantes.fr}

\date{\today}

\begin{document}
\maketitle
\begin{abstract} 

Inspired by a paper by T. Chakradhar, K. Gittins, G. Habib and N. Peyerimhoff,  we analyze their conjecture that the ground state energy of the magnetic Dirichlet-to-Neumann operator tends to infinity as the magnetic field tends to infinity.   More precisely,  we prove refined conjectures 
 for general two dimensional domains,  based on the 
analysis in the case  of the half-plane and  the disk  by two of us  (B.H.  and F.N.).
We also extend our analysis to  the three dimensional case,  and explore a connection with the eigenvalue asymptotics of the magnetic Robin Laplacian.
 \end{abstract}

 
 \vspace{0.5cm}
 

 

\section{Introduction}

\subsection{The magnetic Dirichlet to Neumann operator.}

This paper is a continuation of \cite{HeNi1} where the ground state energy of the Dirichlet to Neumann operator in the case with a constant magnetic field in the unit disk $D(0,1) \subset \R^2$ was studied.
Here, we extend this problem in the case of general bounded  domains in $\R^n$ ($n=2,3$).

\vspace{0.1cm}\noindent

Let  $\Omega$ be a  bounded connected subset of $\R^n$,  with smooth boundary $\partial\Omega$  consisting of a finite number of connected components; in short we say that  $\Omega$ is a regular domain of $\R^n$.   For any $ u \in \mathcal D'(\Omega)$, the magnetic Schr\"odinger operator  on $\Omega$ is defined as 
\begin{equation}\label{defMagOp}
	H_{A}\  u = (D-A)^2 u = -\Delta u -2i \   A \cdot \nabla u  + (|A|^2 -i \  {\rm{div}} \ A ) u,
\end{equation}
where   $D= -i \nabla$,  $-\Delta$ is the usual positive Laplace operator on $\R^n$ and ${\displaystyle{A= \sum_{j=1}^n A_j dx_j}}$ is the 1-form magnetic potential. We often identify the $1$-form magnetic potential $A$ with the vector field 
$\overrightarrow{A} = (A_1, ..., A_n)$, and we assume that $\vec{A} \in C^{\infty}(\overline{\Omega}; \R^n)$. The magnetic field is given by the $2$-form $B =dA$.  We will also use the notation $\vec{H} =\curl \vec{A}$.

\vspace{0.1cm}\noindent
Since zero does not belong to the spectrum of the Dirichlet realization  of  $H_{A}$,  the boundary value  problem 
\begin{equation}  \label{Dirichlet}
	\left\{
	\begin{array}{rll}
		H_{A} \  u &=&0  \  \ \rm{in}  \ \ \Omega,\\
		u_{  \vert \partial \Omega} & =& f \in H^{1/2}(\partial\Omega) ,
	\end{array}\right.
\end{equation}
has a unique solution $u \in H^1(\Omega)$, that we call the {\it{magnetic harmonic extension}} of $f$. The  Dirichlet to Neumann map, (in what follows D-to-N map), is defined  by
 \begin{equation} \label{D-to-N--map}
	\begin{array}{rll}
		\Lambda_{A} :   H^{1/2}(\partial\Omega)& \longmapsto & H^{-1/2}(\partial\Omega) \\ 
		f   &\longmapsto&  \left(\partial_{\nu} u + i \langle A, \vec{\nu} \rangle \ u  \right)_{  \vert \partial \Omega} ,
	\end{array}
\end{equation}
where $\vec{\nu}$ is the outward normal unit vector field on $\partial\Omega$.  More precisely, we define the D-to-N map using the equivalent weak formulation : 
\begin{equation}\label{DtNweak}
	\left\langle \Lambda_{A} f , g \right \rangle_{H^{-1/2}(\partial \Omega) \times H^{1/2}(\partial \Omega)} = \int_\Omega  \langle (-i\nabla -A)u , (-i\nabla -A)v \rangle\ dx\,,
\end{equation}
for any $g \in H^{1/2}(\partial \Omega)$ and $f \in H^{1/2}(\partial \Omega)$ such that $u$ is the unique solution of (\ref{Dirichlet}) and $v$ is any element of $H^1(\Omega)$ so that $v_{|\partial \Omega} = g$. Clearly, the D-to-N map is a positive operator.

\vspace{0.1cm}\noindent
We recall that since $\Omega$ is assumed to be bounded  the spectrum of the D-to-N operator is discrete and is given by an increasing sequence of eigenvalues 
\begin{equation} \label{spectrum}
	0 \leq 	\mu_1 \leq \mu_2 \leq ... \leq \mu_n \leq ...  
\end{equation}
which tends to $+\infty$.  

\vspace{0.4cm}\noindent
Due to the identity in \eqref{DtNweak}, the D-to-N lowest eigenvalue $\mu_1 :=\lambda^{\rm DN}_1(A,\Omega)$   can be expressed in the variational form as
\begin{equation}\label{eq:def-var-DN}
\lambda^{\rm DN}_1(A,\Omega)=\inf_{u\in C^\infty(\overline\Omega),\ \|u\|_{\partial\Omega} =1} \|(-i\nabla- A)u\|^2_\Omega\,,
\end{equation}
and the eigenvalues $\mu_j:=\lambda^{\rm DN}_j(A,\Omega)$ can be expressed by the min-max principle, 
\begin{equation}\label{eq:def-var-DN-j}
\lambda^{\rm DN}_j(A,\Omega)=\inf_{\substack{M\subset C^\infty(\overline\Omega)\\\mathrm{dim}(M)=j}}\Biggl(\max_{u\in M, \ \|u\|_{\partial\Omega} =1} \|(-i\nabla- A)u\|^2_\Omega\Biggr),
\end{equation}
where $\| \cdot \|_\Omega$ and $\|\cdot \|_{\partial\Omega}$ denote the $L^2$-norms in $L^2(\Omega;\C)$ and $L^2(\partial\Omega;\C)$ respectively.

\vspace{0.2cm}\noindent
For $j=1$, we simply write $\lambda^{\rm DN}(A,\Omega)$ instead of $\lambda^{\rm DN}_1(A,\Omega)$.

\subsection{Planar domains}
In the two dimensional case,   we prove accurate asymptotics for the lowest eigenvalue  of the magnetic D-to-N operator.

\vspace{0.2cm}\noindent
Our first result concerns the constant magnetic field.

\vspace{0.1cm}\noindent
\begin{theorem}\label{conj:HN}
 	Let $\Omega$ be a regular domain in $\mathbb R^2$ and $A$ be a magnetic potential with constant magnetic field with norm $1$.  Then the ground state energy of the D-to-N map $\Lambda_{bA}$ satisfies
 \begin{equation}\label{eq:conj-HN}
 \lim_{b\rightarrow +\infty} b^{-1/2} \; \lambda^{\rm DN}(b A,\Omega) = \hat{\alpha}:= \frac{\alpha}{\sqrt{2}},
 \end{equation}
 where $-\alpha$ is the unique negative zero of the parabolic cylindrical function $D_{1/2}(z)$.
\end{theorem}
\vspace{0.1cm}\noindent
This theorem was conjectured in \cite{HeNi1}.
\vspace{0.2cm}\noindent
We recall that, for $\nu \in \mathbb R$, the parabolic cylinder function $D_\nu (z)$ is the (normalized)  solution of the differential equation 
\begin{equation}\label{ODEDnu0}
	w'' + \Bigl(\nu  + \frac 12 - \frac{z^2}{4}\Bigr)\ w=0\,,
\end{equation}
which tends to $0$ as $z \to +\infty$.  More precisely, $D_\nu (z)$ has the following asymptotic expansion 
\begin{equation}\label{asymptDnu}
D_{\nu}	 (z) = e^{-\frac{z^2}{4}} z^\nu \left( 1 + \mathcal O \big(\frac{1}{z^2}\big) \right) \ ,\ z \to + \infty\,.
\end{equation}

\vspace{0.1cm}\noindent
We refer to (\cite{HeNi1},  Section 2) for more details on the parabolic cylinder functions. At last, the positive real $\alpha$ appearing in Theorem \ref{conj:HN} is approximately equal to 
\begin{equation}\label{approalpha}
	\alpha = 0.7649508673 ....
\end{equation}

\vspace{0.3cm}\noindent
One can actually get a two-terms asymptotics where the second term takes account of the curvature  of the boundary.  The following result is a generalization of \cite[Theorem 1.1]{HeNi1},  in the case of the disk. 
\begin{theorem}
\label{thm:main-2D}
 	Let $\Omega$ be a regular domain in $\mathbb R^2$ and $A$ be a magnetic potential with constant magnetic field with norm $1$.  Then the ground state energy of the D-to-N map  $\Lambda^{DN}_{bA}$ satisfies
 \begin{equation*}
\lambda^{\rm DN}(b A,\Omega) =  \hat \alpha b^{\frac 12}  -   \frac{\hat \alpha^2 +1}{3}  \max_{x\in \partial \Omega} \kappa_x \, + o(1)\,,
 \end{equation*}
 where $\kappa_x$ denotes the curvature at $x$.
\end{theorem}
We actually prove that the asymptotics in Theorem~\ref{thm:main-2D} holds for the $j$'th eigenvalue,  for every fixed $j\geq 2$,  and when the magnetic field is  only supposed to be constant  in a neighborhood of $\partial\Omega$. If $\Omega$ is simply connected,  the following inequality holds \cite{Pa}
\[\max_{x\in \partial \Omega} \kappa_x\geq \sqrt{\pi/|\Omega|} \,,\]
and we obtain as corollary of Theorem~\ref{thm:main-2D}:
\begin{corollary}\label{corol:main2D}
Let $\Omega$ be a regular domain in $\mathbb R^2$ and $A$ be a magnetic potential with constant magnetic field with norm $1$.  Suppose that $\Omega$ is simply connected and $\cB$ is a disk with the same area as $\Omega$. Then there exists $ b_1(\Omega)>0$ such that,  for all $b\geq  b_1(\Omega)$, 
 \[
\lambda^{\rm DN}(b A,\Omega) \leq   \lambda^{\rm DN}(b A,\cB).
 \]
\end{corollary} 
The inequality in Corollary~\ref{corol:main2D} is reminiscent of an inequality for the magnetic Laplacian \cite{FH},  and it would be interesting to investigate whether it holds for all $b>0$.  For the magnetic Laplacian,  there is progress in the study of this question 
\cite{CLPS, KL, KL2}.  The geometric isoperimetric inequality also yields that $\lambda^{\rm DN}(b A,\Omega) \leq   \lambda^{\rm DN}(b A,\cB_*)$, where $\cB_*$ is a disk with the same perimeter as $\Omega$.\medskip

If the magnetic field is variable, not vanishing in $\overline{\Omega}$ and constant along the boundary one could expect a more general result in the spirit of the one  of N. Raymond \cite{Ray1} 
devoted to the ground state energy of the Neumann magnetic Laplacian. The second term will also involve the normal derivative of the magnetic field to the boundary.\medskip

 We will also consider the case of variable  magnetic field  in $2D$ and in $3D$ in the same spirit as for the analysis of the Neumann problem appearing in surface superconductivity  \cite{LuPa0,HeMo,Ray1,Ray2,HeMo3}.

\vspace{0.1cm}\noindent
We prove in particular the following  theorem (we refer to \cite{HeMo1, HeMo3,LuPa} for the Neumann problem). 
\begin{theorem}
\label{conj4}
 	Let $\Omega$ be a regular domain in $\mathbb R^2$, $A$ be a magnetic potential  with non vanishing  magnetic field $B(x)$ in $\overline \Omega$, then the ground state energy of the D-to-N map $\Lambda_{b A}$ satisfies
 \begin{equation}
\lambda^{\rm DN}(b A,\Omega) = \hat \alpha \bigl(\inf_{x\in \partial \Omega} |B(x)|\bigr)^{\frac 12} b^\frac 12 + o(b^\frac 12)\,.
 \end{equation}
\end{theorem}

\begin{remark}
 We will prove a more general result valid for a larger class of magnetic fields and for the low-lying eigenvalues.   Actually,  Theorem \ref{conj4} is still true if the magnetic field $B$  does not vanish  on $\partial \Omega$ and if  the set $\mathcal Z(B):=\{x\in\Omega\colon B(x)=0\}$ consists of a finite number of smooth curves such that $|\nabla B|>0$ on $\mathcal Z(B)$. See  Assumption \ref{ass:adm-B} and 
Example \ref{ex:adm-B} for other conditions. \\
 Interestingly,  only the values of the magnetic field on the boundary contributes   to the main term in the asymptotics for the D-to-N operator.  In fact,  unlike the Neumann magnetic Laplacian,  there is no contribution involving $\inf_{x\in \Omega}|B(x)|$.
\end{remark}

\vspace{0.1cm}\noindent

\subsection{Three dimensional case}
\vspace{0.1cm}\noindent
We have a similar result  for variable magnetic fields in $3$D which is in correspondence with known results obtained in the analysis of the ground state 
energy of the Neumann realization of the magnetic Laplacian (see \cite{LuPa,HeMo3,Ray2}):

\begin{theorem}\label{thm:main-3D}
	Let $\Omega$ be a regular bounded domain in $\mathbb R^3$, $A$ be a magnetic potential with non vanishing  magnetic field  $B(x)$ in $\overline{\Omega}$, then the ground state energy of the D-to-N map $\Lambda^{DN}_{bA}$ satisfies
	\begin{equation}
		\lim_{b\rightarrow +\infty} b^{-1/2} \; \lambda^{\rm DN}(b A,\Omega) =   \inf_{x\in \partial \Omega} \Big(\lambda^{\rm DN}(\vartheta(x))|B(x)|^{\frac 12}\Big)\,,
	\end{equation}
	where, for $x\in \partial \Omega$, 
	\begin{itemize}
	\item   $\vartheta(x)$ is defined by 
	\begin{equation}\label{ortha}
		\langle \vec{H}(x)\,|\,  \vec{\nu}  \rangle= - |B(x)|  \sin \vartheta (x)\;.
	\end{equation}
	\item  $\vec{H}(x)$ is the magnetic vector field associated with $B(x)$ considered as a $2$-form by the Hodge-map. \item
	 $\vec{\nu}$ is the exterior normal at $x\in \partial \Omega$.
	 \item  $\lambda^{\rm DN}(\vartheta)$ is the ground state energy (see \eqref{eq:6.2}) relative to the half space when the magnetic field is constant.
	 \end{itemize}
\end{theorem}

 There are two important consequences of Theorem~\ref{thm:main-3D}:
\begin{itemize}
\item When $B$ is constant with magnitude  $1$,  it follows   that
\[ \lim_{b\rightarrow +\infty} b^{-1/2} \; \lambda^{\rm DN}(b A,\Omega) =\inf_{x\in \partial \Omega} \lambda^{\rm DN}(\vartheta(x)),\]
which is consistent with Theorem~\ref{conj:HN} for 2D domains.  
\item More generally,  if we know only that $|B(x)|$ is constant,  as for the helical magnetic field $B(x)=(\cos(\tau x_3),\sin(\tau x_3),0)$ encountered in liquid crystals \cite{Pan, HK},    then  
\[ \lim_{b\rightarrow +\infty} b^{-1/2} \; \lambda^{\rm DN}(b A,\Omega) =|B|^{\frac 12} \inf_{x\in \partial \Omega} \lambda^{\rm DN}(\vartheta(x))\,.\]
\end{itemize} 
Notice, that if $\partial \Omega$ has a component homeomorphic to the sphere $\mathbb S^2$, then the hairy ball Theorem applied to the tangential part of the magnetic field at the boundary implies that there exists a point $x\in \partial \Omega$ 
such that $\vartheta(x) =0$ and we deduce
\[
\inf_{x\in \partial \Omega} \lambda^{\rm DN}(\vartheta(x))=\hat \alpha\,.
\]
\subsection{The magnetic Robin Laplacian}

 We can  get  information about $\lambda^{\rm DN}(bA,\Omega)$ by comparing with the lowest eigenvalue of a Robin problem
\begin{equation}\label{eq:int-R-ev}
\mu(\hat\lambda)=\mu(bA,\hat\lambda,\Omega)=\inf_{u\not=0}\frac{\|(-i\nabla-bA)u\|_\Omega^2-\hat\lambda\|u\|_{\partial\Omega}^2}{\|u\|^2_\Omega}.
\end{equation}
In two dimensions and under constant magnetic field,  
two-term  and three-term asymptotics for $\mu(\hat\lambda)$ are available  \cite{Ka06,  FLRV} in the regime where $b\to+\infty$ and
\[\hat\lambda=\hat\alpha b^{1/2}+o(b^{1/2}).\] 
If we choose 
\[\hat\lambda^*=\hat\alpha b^{1/2}+o(b^{1/2})\mbox{ such that }\mu(\hat\lambda^*)\leq 0,\]
and
\[\hat\lambda_*=\hat\alpha b^{1/2}+o(b^{1/2})\mbox{ such that }\mu(\hat\lambda_*)\geq 0,\]
then, we get using the characterization in \eqref{eq:def-var-DN}, 
\[\hat\lambda_*\leq\lambda^{\rm DN}(bA,\Omega)\leq \hat\lambda^*.\]
We will use this approach in Section~\ref{sec:splitting} to analyze the splitting of the low-lying eigenvalues under a constant magnetic field and we  obtain
\begin{theorem}\label{thm:splitiing*}
 Suppose that $\Omega$ is a regular domain in $\R^2$ such that the curvature  of the boundary has a unique non-degenerate maximum.   Let $A$ be a vector field on $\Omega$ with constant magnetic field $\curl A=1$.   Then there is a constant $ K_*(\Omega)>0$ such that
    \[ \lambda_2(bA,\Omega)-\lambda_1(bA,\Omega)=K_*(\Omega) b^{-1/4}+o(b^{-1/4}) \quad\mbox{ as }b\to+\infty.\]
\end{theorem}

The Robin problem is not analyzed under non-constant magnetic field in two dimensions.  For the three dimensional case,  the existing results in \cite{HKR} only cover the regime $|\hat\lambda|=o(b^{1/2})$,  whereas the relevant regime for  our setting is $|\hat\lambda|\propto b^{1/2}$.  

\subsection{Weak field limit}

Our final result concerns the limit as $b\to0$ in simply connected domains.   Assuming that $\Omega\subset\R^n$ is simply connected,  $n=2,3$,  there is a unique vector field $A_\Omega\in H^1(\Omega;\R^n)$ such that \cite[Prop. D.1.1]{FH4}
\[\curl A_\Omega=1\mbox{ and } \mathrm{div}\,A_\Omega=0\mbox{ on }\Omega,\quad \vec{\nu}\cdot A_\Omega=0\mbox{ on} \ \partial\Omega.\] 

\begin{theorem}\label{thm:small-b}
Let $\Omega$ be a regular domain in $\R^n$ and $A$ be a vector field generating a constant magnetic field $\curl A=1$.  Then, the lowest eigenvalue of the D-t-N operator in $\Omega$ satisfies
\[ \lambda^{\rm DN}(bA,\Omega)=\frac{b^2}{|\partial\Omega|} \int_\Omega|A_\Omega|^2dx+o(b^2)\quad\mbox{ as } b\to0. \]
\end{theorem}
Compared with the magnetic Laplacian,  the coefficient of $b^2$ is  the average of $|A_\Omega|^2$ over $\Omega$ \cite[Proposition~1.5.2]{FH4}.   

By \cite[Proposition 3.1]{FH} and the geometric isoperimetric inequality,  we obtain as corollary of Theorem~\ref{thm:small-b}:
\begin{corollary}\label{corol:small-b}
Let $\Omega$ be a regular domain in $\mathbb R^2$ and $A$ be a magnetic potential with constant magnetic field with norm $1$.  Suppose that $\Omega$ is simply connected and $\cB$ is a disk with the same area as $\Omega$. Then there exists $b_0(\Omega)>0$ such that,   $|b|\leq b_0(\Omega)$, 
 \[
\lambda^{\rm DN}(b A,\Omega) \leq   \lambda^{\rm DN}(b A,\cB).
 \]
\end{corollary}

\begin{remark}
 When $\Omega$ is the disk of radius $R$, the magnetic potential $A_\Omega (x,y) = \half (-y,x)$ and the lowest eigenvalue of the D-t-N operator is explicitly given by (see \cite[Remark 6.1]{HeNi0} and \cite[Remark 5.10]{HeNi1}), 
	\begin{equation*}
	\lambda^{\rm DN}(bA,D(0,R))= \frac{bR}{2} \ \frac{ I_0' (\frac{bR^2}{4})} {I_0 (\frac{bR^2}{4}) }\,,
	\end{equation*}
	where 	
		\begin{equation}\label{BesselI0}
			I_0(z) =\sum_{k=0}^{+\infty}  \frac{z^{2k}}{2^{2k} (k!)^2}
		\end{equation} 
		is the modified Bessel function of the first kind of order $0$. So, in this case, we can give a more accurate asymptotics: 
		\begin{equation}\label{lambda0}
		\lambda^{\rm DN}(bA,D(0,R))= \frac{R^3 b^2}{16} + O(b^4).
		\end{equation}
\end{remark}

\subsection{Organization}

In Section~\ref{sec:pre},  we collect preliminaries to be used throughout the paper.  In Sections~\ref{sec:2D-ub} and \ref{sec:2D-lb},  we prove upper and lower bounds on the low-lying eigenvalues in the two dimensional case,  which yield  Theorems~\ref{conj:HN},  \ref{thm:main-2D} and \ref{conj4}.  For generic 2D domains, we study the splitting of the low-lying eigenvalues in Section~\ref{sec:splitting}. The proof of Theorem~\ref{thm:main-3D} in $3$D domains occupies Section~\ref{sec:3D}.  We  prove Theorem~\ref{thm:small-b}. in Section~\ref{sec:small-b}.  Finally,  there are two appendices on gauge transformations and a reference operator in the half-space,  respectively.

\section{Preliminaries}\label{sec:pre}

\subsection{The half-plane}
In the half-plane $\R_+^2=\{(t,x)\colon t>0\}$, for $b\in\R$, we introduce
\begin{equation}\label{eq:HN-m(b)}
E(b)=\inf_{u\in C^\infty(\overline{\R_+^2}), \|u\|_{\partial\mathbb R^2_+}=1} \|(-i\nabla- bA_{\tau})u\|^2_{\R_+^2}\,,
\end{equation}
where $A_{\tau}(t,x)=(0,t)$, for $(t,x)\in\R_+^2$. Notice that $A_{\tau}$ is tangent to the boundary of $\R_+^2$, and it generates a constant unit magnetic field with norm $1$.

\vspace{0.1cm}\noindent
The sign of $b$ is irrelevant by the invariance under the unitary transformation of complex conjugation, since
\begin{equation}\label{eq:E(-b)}
E(b)=E(-b).
\end{equation}
By scaling (see \cite[Section 3]{HeNi1}),  
 $E(b)$ and $\hat \alpha$ can be expressed as\footnote{In \cite{HeNi1}, the authors consider $m(b)=E(2b)$ and $\alpha=m(1)$.}
\begin{equation}\label{eq:HN-sec3}
E(b)=b^{1/2}E(1),\quad \hat \alpha=E(1).
\end{equation}

\subsection{Harmonic oscillator with Robin condition}\label{sec:1D-Robin}
The constant $\hat\alpha$ is also  related to the harmonic oscillator on the half-axis
\[-\frac{d^2}{dt^2} +(t-\xi)^2\quad\mbox{on }\R_+, \]
with Robin boundary condition at $t=0$,  (i.e with the boundary condition  $u'(0)=\gamma u(0)$). The lowest eigenvalue of this operator was studied in \cite{LuPa99, Ka06}, and  the  other eigenvalues are more recently  studied  in \cite{FLRV}.

\vspace{0.1cm}\noindent
We denote by $\|\cdot\|_2$ the $L^2$-norm on $\R_+$, and for  
$b>0$, $\gamma,\xi\in\R^2$, we introduce the lowest eigenvalue  of the Robin magnetic harmonic oscillator:
\[\mu(\gamma,\xi;b)=\inf_{\|f\|_2=1}\Bigl(\int_{0}^{+\infty}\bigl(|f'(t)|^2+(bt-\xi)^2|f(t)|^2\bigr)dt+\gamma|f(0)|^2\Bigr).\]
If we minimize over $\xi\in\R$, we set
\[\Theta(\gamma;b)=\inf_{\xi\in\R}\mu(\gamma,\xi;b),\]
and we have by scaling
\begin{equation}\label{eq:DG-scaling}
\Theta(\gamma;b)=b\Theta(b^{-1/2}\gamma;1),\quad \mu(\gamma,\xi;b)=b\mu(b^{-1/2}\gamma,b^{-1/2}\xi;1).
\end{equation}
This shows that it suffices to consider the case $b=1$, and we thus introduce
\begin{equation}\label{eq:1D-Robin*}
\Theta(\gamma):=\Theta(\gamma;1),\quad \mu(\gamma,\xi) :=
\mu(\gamma;\xi,1).
\end{equation}
 The Neumann case $\gamma=0$ corresponds to the de\,Gennes model, and  $$\Theta_0=\Theta(0)\,$$ is the so-called  de\,Gennes constant.   It is known that approximatively $\Theta_0\approx 0.590106$.

\vspace{0.1cm}\noindent
As function of $\xi \in ]-\infty, +\infty[$, the eigenvalue $\mu(\gamma,\xi)$ decreases from $+\infty$ until it reaches  a unique minimum  attained at 
\begin{equation}\label{eq:th-xi}
\xi(\gamma)=\sqrt{\Theta(\gamma)+\gamma^2},
\end{equation}
then  $\mu(\gamma,\xi)$ increases to $1$. The function $\R\ni\gamma\to\Theta(\gamma)\in(-\infty,1)$ is smooth and increasing \cite[Theorem~II.1 and Proposition~II.5]{Ka06}, and it has a unique zero $\gamma_0<0$,    (see Figure 1):
\begin{equation}
\Theta(\gamma_0)=0\,.
\end{equation}
Furthermore, the derivative of $\Theta(\cdot)$ is given in \cite[Proposition~II.5]{Ka06} as
\begin{equation}\label{eq:der-Theta}
\Theta'(\gamma)=|u_\gamma(0)|^2,
\end{equation}
where $u_\gamma$ is a  eigenfunction corresponding to the eigenvalue $\Theta(\gamma)$, normalized in $L^2(\R_+)$.

\vspace{0.1cm}\noindent
We can also express the relation between $\xi(\gamma)$ and $\Theta(\gamma)$ in terms of the parabolic cylindrical functions as in \cite[Eqs. (2.41)-(2.42)]{HeNi1}. In fact, one can prove that $\xi(\gamma)$ satisfies the implicit equation:
\[
\sqrt{2}\  \ D'_{\frac{ \Theta(\gamma) -1}{2}} ( -\sqrt{2}\ \xi(\gamma)) = \gamma \
D_{\frac{\Theta(\gamma) -1}{2}} ( -\sqrt{2}\ \xi(\gamma))\,.
\] 
Under the constraint (\ref{eq:th-xi}), or equivalently using the relation (\cite{MOS1966}, p. 327),
\begin{equation}
D'_\nu(z) - \frac{z}{2} D_\nu (z) + D_{\nu+1} (z)=0\,,
\end{equation}	
we see that $\xi(\gamma)$ is also solution of
\begin{equation} \label{otherelation}
-\sqrt{2} D_{\frac{\Theta(\gamma)+1}{2}}(-\sqrt{2} \xi(\gamma))	= (\gamma+ \xi(\gamma))  \ D_{\frac{\Theta(\gamma)-1}{2}}(-\sqrt{2} \xi(\gamma))\,.
\end{equation}
In light of \eqref{otherelation}, we have the following:
\vspace{0.1cm}
\begin{itemize}
    \item If we take  $\gamma=0$, we recover  \cite[Eqs (2.40)-(2.41)]{HeNi1}.
    \item If we take $\gamma=\gamma_0$, knowing that $\xi(\gamma_0)=-\gamma_0$ and $\Theta(\gamma_0)=0$, we derive obviously $D_{\half} (\sqrt{2}\,\gamma_0)=0$ in (\ref{otherelation}). Since $\hat\alpha=\alpha/\sqrt{2}$ and $-\alpha$ is the unique negative zero of $D_{\half}$, we finally get that 
    \begin{equation}
    \hat\alpha=-\gamma_0\,.
    \end{equation}
\end{itemize}
	\begin{figure}[h]
	\begin{center}
		\includegraphics[width=0.32\textwidth]{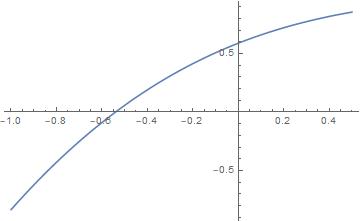}
	\end{center}
	\caption{Graph of the function $\Theta(\gamma)$. }
\end{figure}
\subsection{D-to-N on the half-axis}~\\
We can also derive the relation $\hat\alpha=-\gamma_0$, where $\gamma_0$ is the unique zero of $\Theta(\cdot)$, directly from the following characterization of $\hat\alpha=\alpha/\sqrt{2}$  (see \cite[Eq. (6.7)]{HeNi1})
\begin{equation}\label{eq:def-alpha-1D}
\hat\alpha=\inf_{\substack{f(0)\not=0\\\xi\in\R}}
\frac{\int_{0}^{+\infty}\bigl(|f'(t)|^2+(t-\xi)^2|f(t)|^2\bigr)dt}{|f(0)|^2}.
\end{equation}
This point of view is helpful to prove the following. 
\begin{proposition}\label{prop:link-1D}
There exists a function $f_*$ in the Schwartz space $\mathcal S(\overline{\R_+})$ such that
\begin{enumerate}[\rm (i)] 
\item $f_*>0$ on $\R_+$ and $f_*(0)=1$.
\item \[\int_{\R_+}\bigl(|f'_*(t)|^2+(t-\hat\alpha)^2|f_*(t)|^2 \bigr)dt=\hat\alpha.\]
\item 
\[
\int_{\R_+}(t-\hat\alpha)|f_*(t)|^2dt=0\,,\quad \int_{\R_+}(t-\hat\alpha)^2|f_*(t)|^2dt=\frac{\hat\alpha}4\,,\quad
\int_{\R_+}(t-\hat\alpha)^3|f_*(t)|^3dt=\frac16(1-2\hat\alpha^2).
\]
\item
\[  \int_{\R_+}f'_*(t)f_*(t)dt=-\frac12\,,\quad  \int_{\R_+}t|f'_*(t)|^2dt=\frac13+\frac{\hat\alpha^2}{12}\,.\]
\end{enumerate}
\end{proposition}
\begin{proof}
Knowing that $\hat\alpha=-\gamma_0$ and $\Theta(\gamma_0)=0$, we infer from \eqref{eq:th-xi} that $\xi( \gamma_0)=\hat\alpha$ and  $\mu(-\hat\alpha,\hat\alpha)=0$. Choose a positive and normalized (in $L^2(\R_+)$) ground state $\varphi_*$ of $\mu(-\hat\alpha,\hat\alpha)$ and put $f_*=(1/|\varphi_*(0)|)\varphi_*$. Then $f_*$ satisfies
\[-f''_*+(t-\hat\alpha)^2f_*=0\mbox{ on }\R_+,\quad f'_*(0)=-\hat\alpha,\quad f_*(0)=1.\]
We get (ii) and $\int_{\R_+}f_*'(t)f_*(t)dt=-\frac12$ by integration by parts. For the  identities in (iii),  we reproduce the calculations in \cite{BS}.
We notice that for  
\[v=2pf_*'-p'f_*\,, \,p(t)=(t-\hat\alpha)^k \mbox{  and } k\in\{0,1,2\}\,,\]
 we have
\[\bigl(-\partial_t^2+(t-\hat\alpha)^2\bigr)v=\bigl(-4(t-\hat\alpha)^2p'-4(t-\hat\alpha)p\bigr)f_*,\]
and we use integration by parts to write
\[\int_{\R_+}\bigl(-4(t-\hat\alpha)^2p'-4(t-\hat\alpha)p\bigr)|f_*(t)|^2dt=\bigl(v'(0)+\hat\alpha v(0)\bigr)f_*(0).\]
The formulas of (iii) follow by considering successively $k=0,1,2$.\\
Finally, the last identity in (iv) is obtained by integration by parts
\[\int_{\R_+} t|f_*'(t)|^2dt=-\int_{\R_+}f_*(t)f_*'(t)dt-\int_{\R_+}tf_*(t)f_*''(t)dt,\]
and by using $f''_*(t)=(t-\hat\alpha)^2f_*(t)$ and the identities in (iii).
\end{proof} 

\begin{remark}\label{rem:0moment}
Although not needed in the proof,  notice that we have
\[
f_*(t)= \frac{D_{-1/2}( \sqrt{2}t - \alpha)}{D_{-1/2}( - \alpha)} \,.
\]
Note that the parabolic cylinder function $D_{-\half}(z)$ is also related with the usual Hankel function:  $D_{-\half}(z) = \sqrt{\frac{z}{2\pi}}\ K_{\frac14}  (\frac{z^2}{4})$. Using Mathematica, we get numerically 
\[
\int_0^{+\infty} f_*(t)^2 dt \approx 0.6861814388\,.
\]
\end{remark}
\subsection{A useful identity}
Recall that $W^{1,\infty}(\Omega;\R)$ is the space of $L^\infty(\Omega)$ real-valued functions $f$ such that $\nabla f\in L^\infty(\Omega)$.

In various proofs, we will use the now standard identity given by  the following proposition.
\begin{proposition}\label{prop:identity}
    Suppose that $A\in H^1(\Omega;\R^2)$. If $u\in H^1(\Omega)$ and $w\in W^{1,\infty}(\Omega,\mathbb R)$, then
    \[
\Re \int_\Omega(-i\nabla - A)u\cdot \overline{(-i\nabla - A) (w^2 u)}\,dx = 
\int_\Omega |(-i\nabla - A) (wu)|^2\,dx-\int_\Omega|\nabla w|^2|u|^2\,dx.\]
\end{proposition}
\begin{proof}
This is a direct consequence of the two identities
\[(-i\nabla-A)(w^2u)=w(-i\nabla-A)(wu)  -i wu \nabla w,\quad w(-i\nabla-A)u=(-i\nabla-A)(wu)+iu\nabla w. \]
\end{proof}
\subsection{Gauge invariance}

If $U\subset\Omega$ is an open set and $\phi:U\to\R$ is in $H^1(U)$, then for any function $u\in H^1(U)$, we have the following identities on $U$,
\[|(-i\nabla-bA)u|^2=|(-i\nabla-bA')v|^2,\quad |u|^2=|v|^2,\quad \curl A= \curl A',\]
where $v=u \,e^{-ib\phi}$ and $A'=A-\nabla\phi$. This amounts to a local gauge transformation. 

\vspace{0.1cm}\noindent
If $U$ is simply connected and we know that the vector potentials $A$ and $A'$ have the same curl on $U$, then we can find a function $\phi$ defined on $U$ such that $A'=A-\nabla\phi$.

\vspace{0.1cm}\noindent
We will use local gauge transformations to transform a given vector potential to a more  convenient one. For instance, if $B=d A$ is not constant, we describe below how we can locally approximate $A$ by a vector potential with constant magnetic field,  up to a local gauge transformation.
\begin{proposition}\label{prop:gauge-gen}
Suppose that $\partial\Omega$ is $C^1$, $A\in H^1(\Omega;\R^2)$ and $B=\curl A$ is $C^1$ on $\overline{\Omega}$. There exists a constant $C>0$ such that, for every $p\in \overline{\Omega}$, there exists a function $\phi_p:U_p\to\R$ such that
\[|A(x)-\nabla\phi_p-B(p)A_0(x-p)|\leq C\, |x-p|^2\quad (x\in U_p), \]
where $A_0(x)=\frac12(-x_2,x_1)$, and $U_p=\overline{\Omega}\cap \cB(p,\delta)$ for some $\delta>0$.
\end{proposition}
\noindent
A stronger version of Proposition~\ref{prop:gauge-gen} is given in \cite[Lemma~3.1]{LuPa0}, when the point $p$ is in $\Omega$.
\begin{proof}[Proof of Proposition~\ref{prop:gauge-gen}]
We can extend $A$ to a compactly supported vector field in  $H^1(\R^2;\R^2)$, and we obtain an extension of $B=\curl A$ to all of $\R^2$ as well. 
In a disk $\cB(p,\delta)$ centered  at $p$, consider the Coulomb gauge
\[A'(x)=2\left(\int_{0}^1 B\bigl(p+s(x-p)\bigr)sds\right)A_0(x-p).\]
Noticing that $\curl A=\curl A'$,   we can write $A=A'-\nabla\phi$ in $\cB(p,\delta)$.  Since $B$ is $C^1$ on $\overline{\Omega}$, we get
\(A'(x)=B(p)A_0(x-p)+\mathcal O (|x-p|^2)\), \(x\in U_p\).
\end{proof}
\subsection{Parallel coordinates in two dimensions}

In the course of the proofs, we will often deal with functions supported in a neighborhood of a boundary point of $\Omega$. In such cases, it is convenient to carry out the computations in \emph{parallel coordinates} that we introduce below.

\vspace{0.1cm}\noindent
Pick a connected component $\Gamma$ of the boundary of $\Omega$, and denote by $L$ its length. By means of the arc-length parametrization, we can identify $\Gamma$ and $\R/L\Z$, where  $s\in \R/L\Z$ is the curvilinear coordinate of a point $x\in\Gamma$. We choose $t_0>0$ sufficiently small such that 
\[\Omega_0:=\{x\in\Omega\colon \mathrm{dist}(x,\Gamma)<t_0\}\] is diffeomorphic to $(0,t_0)\times (\mathbb R/L\mathbb Z)$. More precisely, we introduce the diffeomorphism 
\begin{equation}\label{eq:def-Phi_0}
\Phi_0:(0,t_0)\times (\mathbb R/L\mathbb Z)\to\Omega_0,
\end{equation}
such that, for $x=\Phi_0(t,s)\in\Omega_0$, $t=\mathrm{dist}(x,\Gamma)$ and $s$  is the curvilinear coordinate of $p(x)$, the the orthogonal projection of $x$ on $\Gamma$.  Thus, $t$ denotes the normal distance to $\Gamma$, and $s$ measures the tangential distance along $\Gamma$.

\vspace{0.1cm}\noindent
Note that, if $\Omega$ is simply connected, the boundary consists of a single connected component, $\Gamma=\partial\Omega$. 
\section{Upper bounds in two dimensional domains}\label{sec:2D-ub}
\subsection{Non vanishing magnetic field}
\begin{proposition}\label{prop:ub-2D}
 	Let $\Omega$ be a regular domain in $\mathbb R^2$ and $A$ be a vector potential with a magnetic field $B=\curl A$ that vanishes nowhere  on $\partial\Omega$.
    Suppose that $B$ is $C^1$ on a neighborhood of $\partial\Omega$. Then, for every fixed $j\geq 1$, the $j$'th eigenvalue value  of the D-to-N map $\Lambda_{b A}$ satisfies
 \begin{equation*}
 \lambda^{\rm DN}_j(b A,\Omega) \leq  \Bigl(\inf_{x\in\partial\Omega}|B(x)|\Bigr)^{\frac12}\hat\alpha\, b^{\frac{1}{2}}+  \mathcal O(b^{1/3})\quad , \quad b\to+\infty \,.
 \end{equation*}
 \end{proposition}
\begin{proof}
Since the D-to-N operators with vector potentials $A$ and $-A$ are unitarily equivalent, it suffices to consider  $B>0$.\medskip

\noindent
{\bf Step 1. The test function.}\medskip

Choose a point $p\in\partial\Omega$ such that $B(p)=\min_{x\in\partial\Omega}B(x)$. By Proposition~\ref{prop:gauge-gen}, we can assume that, modulo a (local) gauge transformation, that $A$ satisfies,
\[A(x)=B(p)A_0(x-p)+\mathcal O (|x-p|^2).\]
Let $\Gamma$ be the connected component of $\partial\Omega$ that contains $p$. Working in parallel coordinates $(t,s)\in(0,t_0)\times(\R/L\Z)$, centered at the point $p$, we introduce the functions
\[v(t,s)=u(t,s)\cdot e^{-ib\varphi(s,t)},\quad  u(t,s)=
 b^{\rho/2}\chi_1(b^{\rho} s) \cdot\chi(b^{\rho}t)\cdot f_*\bigl(B(p)^{1/2}b^{1/2}t\bigr)\cdot e^{-i(B(p)b)^{1/2}\hat\alpha s}, \]
where $\chi_1\in C_c^{\infty}(-1,1)$ is normalized in $L^2(\R)$, $\chi$ is a smooth cut-off function, equal to $1$ in a neighborhood of $0$, $\rho\in(0,\frac12)$, and $f_*$ is the function introduced in Proposition~\ref{prop:link-1D}. The function $\varphi$ is real-valued and amounts to a (local) gauge transformation in the $(t,s)$ coordinates (its choice will be explained below). 

\vspace{0.1cm}\noindent
With $\Phi_0$ the diffeomorphism introduced in \eqref{eq:def-Phi_0}, $v\circ\Phi_0$ defines a function in $H^1(\Omega)$, which will be the test function with which we will work. Since $\chi_1$ is nomalized in $L^2(\R)$, we get that  the restriction of 
$v\circ\Phi_0$  to $\partial \Omega$ is normalized in $L^2(\partial\Omega)$.
\vspace{0.1cm}\noindent
Put $A^{\rm lin}(x):=B(p)A_0(x-p)$. We choose the function $\varphi$ such that (see Appendix~\ref{appendix.A}):
\begin{multline}\label{eq:ub-qf}
\|(-i\nabla-bA^{\rm lin})v\circ\Phi_0\|^2_\Omega=\\
\int_{\R_+^2}\Bigl(|\partial_tu|^2+(1-tk(s))^{-2}\bigl|\bigl(-i\partial_s+B(p)b(t-\mbox{$\frac12$}t^2k(s)\bigr)u\bigr|^2 \Bigr)(1-tk(s))dtds,
\end{multline}
where $k(s)$ is the curvature at the point of curvilinear coordinate $s$.\medskip

\noindent
{\bf Step 2. Some estimates.}\medskip

Our choice of the  function $u$ yields
\[ \int_{\R^2_+}|u(t,s)|^2dtds\leq \bigl(B(p)b\bigr)^{-1/2}\int_{\R_+}|f_*(\tau)|^2d\tau.\]
With $\epsilon\in(0,1)$, we have by the Cauchy-Schwarz inequality,
\[\int_{\R^2_+}\bigl|\bigl(-i\partial_s+B(p)b(t-\mbox{$\frac12$}t^2k(s)\bigr)u\bigr|^2dtds
\leq (1+\epsilon)\int_{\R^2_+}\bigl|(-i\partial_s+B(p)bt)u\bigr|^2dtds+\mathcal O(\epsilon^{-1}b^{\frac32-4\rho}).\]
Moreover, a routine calculation yields
\[\int_{\R^2_+}\bigl|(-i\partial_s+B(p)bt)u\bigr|^2dtds\leq (1+\epsilon)\bigl(B(p)b\bigr)^{1/2}\int_{\R_+}(\tau-\hat\alpha)^2|f_*(\tau)|^2d\tau+\mathcal O(\epsilon^{-1}b^{-\frac12+2\rho}). \]
Using that $f_*$ is a   Schwartz function, we have\footnote{We write $F=\mathcal O(b^{-\infty})$ if, for any $n\in\mathbb N$, $F=\mathcal O(b^{-n})$ as $b\to+\infty$.}
\[\int_{\R^2_+}|\partial_tu|^2dtds=\bigl(B(p)b\bigr)^{1/2}\int_{\R_+^2} |f'_*(\tau)|^2d\tau+\mathcal O(b^{-\infty}).\]
Returning back to \eqref{eq:ub-qf}, 
 we bound $1-tk(s)$ from above by $1+\mathcal O(b^{-\rho})$. Using ii. in Proposition~\ref{prop:link-1D},  and choosing $\epsilon=b^{-\rho}$ and $\rho=1/3$, 
 we get 
\[\|(-i\nabla-bA^{\rm lin})v\circ\Phi_0\|^2_\Omega\leq \bigl(B(p)b\bigr)^{1/2}\hat\alpha+\mathcal O(b^{1/3}).
\]

\vspace{0.1cm}\noindent
{\bf Step 3. Finishing the proof for the lowest eigenvalue.}\medskip

Using the Cauchy-Schwarz inequality, one has 
\begin{align*}
    \|(-i\nabla-bA)v\circ\Phi_0\|^2_\Omega&\leq (1+b^{-1/6})\|(-i\nabla-bA^{\rm lin})v\circ\Phi_0\|^2_\Omega+
    \mathcal O\bigl(b^{1/6}\|b(A-A^{\rm lin})v\circ\Phi_0\|^2_\Omega \bigr)\\
    &\leq \bigl(B(p)b\bigr)^{1/2}\hat\alpha+\mathcal O(b^{1/3}).
\end{align*}
Since $v\circ\Phi_0^{-1}$ is normalized in $L^2(\partial\Omega)$, and $B(p)$ is the minimum of $B$ on the boundary, this finishes  the proof of the proposition for $j=1$, thanks to the characterization of $\lambda^{\rm DN}_1(bA,\Omega)$ in \eqref{eq:def-var-DN}.\medskip

\vspace{0.1cm}\noindent
{\bf Step 4. Finishing the proof for the $j$'th eigenvalue, $j>1$.}\medskip

Consider $\chi_1,,\cdots,\chi_j\in C_c^\infty(-1,1)$ that constitute an orthonormal set in $L^2(\R)$ and such that their supports are pairwise disjoint.  We slightly modify the test function by introducing, for every $j\in\N$,
\[v_j(t,s)=u_j(t,s)\cdot e^{-ib\varphi(s,t)},\quad  u_j(t,s)=b^{\rho/2}\chi_{j}(b^{\rho}s)\cdot\chi(b^{\rho}t)\cdot f_*\bigl(B(p)^{1/2}b^{1/2}t\bigr)\cdot e^{-i(B(p)b)^{1/2}\hat\alpha s}. \]
Then, with $\rho=\frac 13$,  and $q\in\{1,\cdots,j\}$,
\[
    \|(-i\nabla-bA)v_q\circ\Phi_0\|^2_\Omega\leq \bigl(B(p)b\bigr)^{1/2}\hat\alpha+\mathcal O(b^{1/3}),
\]
and for $q'\not=q$,
\[
  \langle (-i\nabla-bA)v_q\circ\Phi_0,(-i\nabla-bA)v_{q'}\circ\Phi_0\rangle_\Omega=0.
\]
Let $M_j=\mathrm{Span}(v_1\circ\Phi_0,\cdots,v_j\circ\Phi_0)$ and $M_j^{\partial \Omega}$ be the space of its restriction to $\partial \Omega$. We observe 
that the  $v_p(0,\cdot)$ ($p=1,\cdots,j$) form an orthonormal basis of 
$M_j^{\partial \Omega}$. 
Hence, $\mathrm{dim}(M_j)=j$,  and 
we  conclude by using the variational formulation in \eqref{eq:def-var-DN-j},
\[ \lambda_j(bA,\Omega)\leq \max_{g\in M_j} \frac{\|(-i\nabla-bA)g\|^2_\Omega}{\|g\|_{\partial\Omega}^2}\leq \bigl(B(p)b\bigr)^{1/2}\hat\alpha+\mathcal O(b^{1/3}).\]
\end{proof}

\par
\subsection{Improvement in the constant magnetic field}
\begin{proposition}\label{prop:ub-2D-B=1}
 	Let $\Omega$ be a regular domain in $\mathbb R^2$ and $A$ be a vector potential with magnetic field $B=\curl A$. Suppose that $B=1$ on a neighborhood of $\partial\Omega$. Then, for every fixed $j\geq 1$, the $j$'th eigenvalue value  of the D-to-N map $\Lambda_{b A}$ satisfies
 \begin{equation*}
 \lambda^{\rm DN}_j(b A,\Omega) \leq  \hat\alpha\, b^{\frac{1}{2}} -  \frac{\hat \alpha^2+1}{3} \max_{x\in\partial\Omega}k(x)+\mathcal O(b^{-1/6})\quad(b\to+\infty),
 \end{equation*}
 where $k$ is the curvature of $\partial\Omega$.
 \end{proposition}
\begin{proof}~\medskip

\vspace{0.1cm}\noindent
{\bf Step 1. The test function.}\medskip

The test function has a  similar structure to the one constructed in Proposition~\ref{prop:ub-2D}, but since the magnetic field is constant on a neighborhood of $\partial\Omega$, we can carry out the computations to sub-leading terms.

Choose a point $p\in\partial\Omega$ such that $k(p)=\max_{x\in\partial\Omega}k(x)$, and let $\Gamma$ be the connected component of $\partial\Omega$ that contains $p$. Let $\Phi_0$ be the coordinate transformation as introduced in \eqref{eq:def-Phi_0}, but we center it at $p$, i.e. $\Phi_0^{-1}(p)=(0,0)$. 

Consider  $\chi_1,\chi_2,\cdots\in C_c^\infty(-1,1)$ that constitute an orthonormal set in $L^2(\R)$ .  For every $j\in\N$, 
the test function has the form $v_j\circ \Phi_0^{-1}$, with $v_j$ defined as
\[v_j(t,s)=u_j(t,s)\cdot e^{-ib\varphi(s,t)},\quad  u_j(t,s)=b^{\rho/2}\chi_j(b^{\rho} s)\cdot\chi(b^{\rho}t)\cdot f_*\bigl(b^{1/2}t\bigr)\cdot e^{-ib^{1/2}\hat\alpha s}. \]
Here, $\chi$ is a cut-off function, equal to $1$ in a neighborhood of $0$, $\rho\in(0,\frac12)$, and $f_*$ is the function introduced in Proposition~\ref{prop:link-1D}.   Moreover,   we can suppose 
that the functions $\chi_1,\cdots,\chi_j$ have pairwise disjoint supports,

The function $\varphi$ is defined in Appendix \ref{appendix.A}.   We can use
 \begin{multline}\label{eq:ub-qf*}
\|(-i\nabla-bA)v_j\circ\Phi_0\|^2_\Omega=\\
\int_{\R_+^2}\Bigl(|\partial_tu_j |^2+(1-tk(s))^{-2}\bigl|\bigl(-i\partial_s+b(t-\mbox{$\frac12$}t^2k(s)\bigr)u_j\bigr|^2 \Bigr)(1-tk(s))dtds.
\end{multline}
Restricting the functions $v_j\circ\Phi_0$ to $\partial\Omega$, we obtain an orthonormal set in $L^2(\partial\Omega)$.\medskip

\vspace{0.1cm}\noindent
{\bf Step 2. Some estimates.}\medskip

Consider $u\in\{u_j\colon j\in\N\}$. Since $f_*$ is a Schwartz function, the function $u$  satisfies, for $m\geq 0$,
\begin{equation}\label{eq:ub-B=1-norms}
\begin{gathered}
\int_{\R^2_+}t^m|u|^2dtds=b^{-\frac{m+1}{2}}\int_{\R_+}\tau^m|f_*(\tau)|^2d\tau+\mathcal O(b^{-\infty}),\\
\int_{\R_+^2} t^m|\partial_tu|^2dtds=b^{-\frac{m-1}{2}}\int_{\R_+} \tau^m|f'_*(\tau)|^2d\tau+\mathcal O(b^{-\infty}).
\end{gathered}
\end{equation}
In the support of $u$, we have $k(s)=k(0)+\mathcal O(b^{-\rho})$. Hence,
\begin{equation}\label{eq:ub-B=1-norms*}
\begin{gathered}
\int_{\R^2_+}t^m|u|^2 k(s) dtds=b^{-\frac{m+1}{2}}k(0)\int_{\R_+}\tau^m|f_*(\tau)|^2d\tau+\mathcal O\bigl(b^{-\frac{m+1}{2}-\rho}\bigr),\\
\int_{\R_+^2} t^m|\partial_tu|^2 k(s) dtds=b^{-\frac{m-1}{2}}k(0)\int_{\R_+} \tau^m|f'_*(\tau)|^2d\tau+\mathcal O\bigl(b^{-\frac{m-1}{2}-\rho}\bigr).
\end{gathered}
\end{equation}
In particular, we have
\begin{equation}\label{eq:dt-term}
\int_{\R_+^2}|\partial_tu|^2(1-tk(s))dtds=b^{1/2}\int_{\R_+}|f_*'(\tau)|^2d\tau
-k(0)\int_{\R_+}\tau |f_*'(\tau)|^2d\tau+\mathcal O(b^{-\rho}).
\end{equation}
{\bf Step 3. More estimates.}\medskip

Since the functions $f_*$ and $\chi$ are real-valued, it is straightforward to verify that
\begin{equation}\label{eq:ds-term} 
\bigl|\bigl(-i\partial_s+b(t-\mbox{$\frac12$}t^2k(s)\bigr)u\bigr|^2=
F+G,
\end{equation}
where
\[ F=b\bigl|\bigl(b^{1/2}(t-\mbox{$\frac12$}t^2k(s))-\hat\alpha\bigr)u\bigr|^2,\quad 
G=b^{ 3\rho} |\partial_s\chi(b^{\rho}s)|^2|\chi(b^{\rho}t)f_*(b^{1/2}t)|^2.\]
Writing $(1-tk(s))^{-1}=\mathcal O(1)$ in the support of $u$, 
and doing the change of variables $(\sigma,\tau)=(b^{\rho}s,b^{1/2}t)$, we get
\begin{equation}\label{eq:ub-B=1-norms**}
\int_{\R_+^2}(1-tk(s))^{-1}G\,dtds=\mathcal O(b^{2\rho-\frac12}).
\end{equation}
To get an accurate estimate of the integral of $(1-tk(s))^{-1}F$, we write  $$(1-tk(s))^{-1}=1+tk(s)+\mathcal O(t^2) \mbox{ and } k(s)=k(0)+\mathcal O(b^{-\rho})$$
 in the support of $u$, and  we expand the square to get
\[\begin{gathered}
(1-tk(s))^{-1}F=(1+tk(0))F+\mathcal O\bigl((b^{-\rho}t+t^2)F\bigr),\\
F=b\bigl[(b^{1/2}t-\hat\alpha)^2-b^{1/2}t^2k(s)(b^{1/2}t-\hat\alpha)+\mathcal O(bt^4)\bigr]|u|^2.
\end{gathered}
\]
Doing a routine calculation, we obtain
\begin{multline}\label{eq:ub=1-norms***}
\int_{\R_+^2}(1-tk(s))^{-1}F\,dtds=b^{1/2}\int_{\R_+}(\tau-\hat\alpha)^2|f_*(\tau)|^2d\tau\\
-k(0)\int_{\R_+}\bigl[\tau^2(\tau-\hat\alpha) -\tau(\tau-\hat\alpha)^2\bigr]|f_*(\tau)|^2d\tau+\mathcal O(b^{-\rho}).
\end{multline}
Returning to \eqref{eq:ds-term}, we infer from \eqref{eq:ub-B=1-norms**} and \eqref{eq:ub=1-norms***},
\begin{multline}\label{eq:ub-B=1-step3}
\int_{\R_+^2}(1-tk(s))^{-1}\bigl|\bigl(-i\partial_s-b(t-\mbox{$\frac12$}t^2k(s)\bigr)u\bigr|^2dtds=b^{1/2}\int_{\R_+}(\tau-\hat\alpha)^2|f_*(\tau)|^2d\tau\\
-k(0)\int_{\R_+}\bigl[\tau^2(\tau-\hat\alpha)-\tau(\tau-\hat\alpha)^2\bigr]|f_*(\tau)|^2d\tau+\mathcal O(b^{-\rho})+\mathcal O(b^{2\rho-\frac12}).
\end{multline}


{\bf Step 4. Finishing the proof.}\medskip

We introduce the following constant
\[C_*=-\int_{\R_+}\tau|f_*'(\tau)|^2d\tau-\int_{\R_+}\bigl[\tau^2(\tau-\hat\alpha)-\tau(\tau-\hat\alpha)^2\bigr]|f_*(\tau)|^2d\tau,\]
which can also be expressed as\footnote{This form is similar to the one in the displayed equation appearing after Eq. (11.27) in \cite{HeMo}.}
\[  C_*= -\int_{\R_+}\tau|f_*'(\tau)|^2d\tau- \hat\alpha \int_{\R_+} (\tau-\hat\alpha)^2 |f_*(\tau)|^2d\tau\,. \]
Then, we collect \eqref{eq:dt-term} and \eqref{eq:ub-B=1-step3} and choose $\rho$ such that $-\rho=2\rho-\frac12$, i.e. $\rho=1/6$. Eventually, for $v\in\{v_j\colon j\in\N\}$, we infer from \eqref{eq:ub-qf*},
\[\|(-i\nabla-bA)v\circ\Phi_0\|^2_\Omega=\hat\alpha b^{1/2}+C_*k(0)+\mathcal O(b^{-1/6}),\]
where we used  (ii)  in Proposition~\ref{prop:link-1D}.

\vspace{0.1cm}\noindent
The space $M_j=\mathrm{Span}(v_n\circ\Phi_0^{-1}\colon 1\leq n\leq j\}$ has dimension $j$ and its restriction to $\partial\Omega$ has an orthonormal basis consisting of the functions $v_p(0,\cdot)$ ($p=1,\cdots,j$).  
Thus, we conclude by \eqref{eq:def-var-DN-j} that 
\[ \lambda^{\rm DN}_j(bA,\Omega)\leq \hat\alpha b^{1/2}+C_*k(0)+\mathcal O(b^{-1/6}).\]
To finish the proof, we recall that $k(0)=\max_{x\in\partial\Omega}k(x)$, and we use (iii) and (iv) in Proposition~\ref{prop:link-1D} to deduce that $$ C_*=-\frac{\hat\alpha^2+1}{3}\,.$$
\end{proof}

\section{lower bounds in two dimensions}\label{sec:2D-lb}

\subsection{Non vanishing magnetic fields}

\begin{assumption}[Admissible magnetic fields]\label{ass:adm-B}
Let $\Omega$ be a regular domain in $\mathbb R^2$ and $A:\Omega\to\R^2$ be a vector potential with a magnetic field $B=\curl A$ such that 
\[ \liminf_{b\to+\infty}\Bigl(b^{-\zeta}\inf_{\|u\|_\Omega=1}\|(-i\nabla-bA)u\|_\Omega^2\Bigr)>0,\]
for some $\zeta>1/2$.
\end{assumption}

\begin{example}
[Magnetic fields satisfying Assumption~\ref{ass:adm-B}]\label{ex:adm-B}~
\begin{enumerate}[\rm (i)]
    \item If $B$ is $C^1$ on $\overline\Omega$ and $|B|>0$ everywhere on $\overline\Omega$, then by \cite{HeMo}, Assumption~\ref{ass:adm-B} holds with $\zeta=1$. 
    \item If $B$ is a non-vanishing step function and the discontinuity set consists of  a finite  number of smooth curves in $\Omega$, then Assumption~\ref{ass:adm-B} holds with $\zeta=1$ (see \cite{A,  AK}).
    \item If $B$ is $C^1$ on $\overline\Omega$, $|B|>0$ on $\partial\Omega$, and the set $\mathcal Z(B)=\{x\in\Omega\colon B(x)=0\}$ consists of a finite number of smooth curves such that $|\nabla B|>0$ on $\mathcal Z(B)$, then by \cite[Theorem~4]{PK}, Assumption~\ref{ass:adm-B} holds with $\zeta=2/3$. 
\end{enumerate}
\end{example}

\begin{proposition}\label{prop:lb}
 	 	Let $\Omega$ be a regular domain in $\mathbb R^2$ and $A$ be a vector potential with a magnetic field $B=\curl A$ that satisfies Assumption~\ref{ass:adm-B}  and does not vanish on $\partial\Omega$. 
    Suppose that $B$  is $C^1$ on   a neighborhood of $\partial\Omega$. 
   Then,  there is $\delta>0$ such that, the ground state energy of the D-to-N map $\Lambda_{b A}$ satisfies
 \begin{equation*}
 \lambda^{\rm DN}_1(b A,\Omega) \geq  \Bigl(\inf_{x\in \partial \Omega} |B(x)|\Bigr)^{\frac 12} b^\frac 12\hat \alpha  + \mathcal O (b^{\frac12-\delta}))\quad , \quad b\to+\infty\,.
 \end{equation*}
\end{proposition}
\begin{proof}
We choose $t_0>0$ such that $B$ is $C^1$ and does not vanish on 
\[\mathcal N:=\{x\in \overline{\Omega}\colon \mathrm{dist}(x,\partial\Omega)<t_0\}.\]
Let $u\in C^\infty(\overline\Omega)$. In the sequel, all estimates will be uniform with respect to $u$, and with respect to $b$ in a neighborhood of $+\infty$. 

\vspace{0.1cm}\noindent
As a consequence of Assumption~\ref{ass:adm-B}, there exist positive constants
$\Theta_1,b_0$ and $\zeta>\frac12$ such that
\begin{equation}\label{eq:HM}
\|(-i\nabla-bA)u\|^2_\Omega\geq \Theta_1 b^{\zeta}\|u\|_\Omega^2\quad\mbox{ for all $b\geq b_0$,}
\end{equation}
and by \cite{HeMo}, 
\begin{equation}\label{eq:HM*}
\|(-i\nabla-bA)u\|^2_\Omega\geq \Theta_1 b\|u\|_\Omega^2\quad\mbox{ if $\supp u\subset \mathcal N$.}
\end{equation}

\noindent
We now  proceed to  the proof of Proposition~\ref{prop:lb}, which we split into several steps.\medskip

\vspace{0.1cm}\noindent
{\bf Step 1:}\medskip

Consider a constant $\rho\in(\frac14,\zeta)$ and a partition of unity 
\[\chi_1^2+\chi_2^2=1\quad\mbox{on }\Omega,\]
where $\chi_1=1$ on $\{\mathrm{dist}(x,\partial\Omega)<b^{-\rho}\} $, $\supp \chi_1 \subset \{\mathrm{dist}(x,\partial\Omega)<2 b^{-\rho}\} $ and
\[ |\nabla\chi_1| +|\nabla\chi_2|=\mathcal O(b^{\rho}).\]
Then, with $\epsilon=b^{-\delta}$ and $\delta>0$, we write
\[\begin{aligned}
 &\|(-i\nabla-bA)u\|^2_\Omega\\
 &=(1-\epsilon)\|(-i\nabla-bA )u\|^2_\Omega+\epsilon\|(-i\nabla-b{A} )u\|^2_\Omega\\
 &=(1-\epsilon)\sum_{j=1}^2\|(-i\nabla-bA )\chi_j u\|^2_\Omega+ \epsilon\|(-i\nabla-b{A} )u\|^2_\Omega+\mathcal O(b^{2\rho}\|u\|^2_\Omega)\\
 &\geq (1-\epsilon)\sum_{j=1}^2\|(-i\nabla-bA )\chi_j u\|^2_\Omega+\bigl({\epsilon\Theta_1 b^\zeta}+\mathcal O(b^{2\rho})\bigr)\|u\|_\Omega^2,
 \end{aligned}\]
 where we have used \eqref{eq:HM}, and Proposition~\ref{prop:identity} with $w=\chi_j$.
 For instance, assuming that 
 \begin{equation}\label{eq:cond-rho1}
2\rho-\zeta+\delta<0\,,
 \end{equation}
  we obtain  that, for $b$ sufficiently large,
\begin{equation}\label{eq:step1}
 \|(-i\nabla-bA )u\|^2_\Omega\geq (1-\epsilon)\|(-i\nabla-bA )\chi_1 u\|^2_\Omega+\frac{\epsilon\Theta_1b^\zeta}{2}\|u\|^2_\Omega.
\end{equation}

\vspace{0.5cm}\noindent
{\bf Step 2:}\medskip

For the sake of simplicity, we suppose that $\Omega$ is simply connected, to ensure that \[\Omega_0:=\{x\in\Omega\colon \mathrm{dist}(x,\partial\Omega)<t_0\}\] is diffeomorphic to $(0,t_0)\times (\mathbb R/L\mathbb Z)$, and use the parallel coordinates defined by the transformation $\Phi_0$ in \eqref{eq:def-Phi_0}. The proof can be easily adjusted to cover the case the non-simply connected case where $\partial\Omega$ consists of a finite number of connected components, by treating doing the computations on each connected component.

\vspace{0.1cm}\noindent
We introduce a partition of unity of $\Omega_0$,
\[\sum_{j}g_j^2=1,\quad \sum_{j}|\nabla g_j|^2=\mathcal O(b^{2\rho}),\]
such that 
\[\mathrm{supp\,}g_j\subset \Phi_0\bigl((0,t_0)\times (s_j-b^{-\rho},s_j+b^{-\rho}\bigr).\]
With  $v=\chi_1u$ supported in $\{|B|>0\}$, we use \eqref{eq:HM*} to write 
\[\|(-i\nabla-bA)v\|^2_\Omega \geq 
(1-\epsilon)\|(-i\nabla-bA)v\|^2_\Omega+\epsilon\Theta_1b\|v\|_\Omega^2.\]
Consequently,  we have, with $v_j=g_j v$,
\begin{equation}\label{eq:step2}
\|(-i\nabla-bA)v\|^2_\Omega \geq 
(1-\epsilon)\sum_{j}\|(-i\nabla-bA )v_j\|^2_\Omega+\epsilon\Theta_1b\|v\|_\Omega^2+\mathcal O(b^{2\rho}\|v\|^2_\Omega).
\end{equation}
For all $j$, pick $p_j\in\partial\Omega\cap\mathrm{supp\,}g_j$ and put $A_j^{\rm lin}(x)=B(p_j)A_0(x-p_j)$. In a neighborhood of $p_j$, we apply a gauge transformation 
\[(v_j,A)\mapsto (v'_j=ve^{-ib\phi_j},A'_j=\nabla\phi_j+A_j^{\rm lin}) \]
as indicated in Proposition~\ref{prop:gauge-gen}. For instance, we have
\[|A-A_j'|\leq C\, b^{-4\rho}\quad \mbox{ on }\mathrm{supp\,}v_j.\]
To lighten the notation, we skip the $'$ when referring to the new configuration $(v_j',A_j')$. We have by the Cauchy-Schwarz inequality,
\[\|(-i\nabla-bA )v_j\|^2_\Omega\geq (1-\epsilon)\|(-i\nabla-bA_j^{\rm lin} )v_j\|^2_\Omega-C\epsilon^{-1}b^{2-4\rho}\|v_j\|^2.\]
Inserting this into \eqref{eq:step2} and arguing as in \eqref{eq:step1}, we get
\begin{equation}\label{eq:step2*}
\|(-i\nabla-bA)v\|^2_\Omega =(1-2\epsilon)\sum_j\|(-i\nabla-bA_j^{\rm lin} )v_j\|^2_\Omega+\frac{\epsilon\Theta_1b}{2}\|v\|^2_\Omega,
\end{equation}
provided that
\begin{equation}\label{eq:cond-rho2} 
1-4\rho+2\delta<0.
\end{equation}
%

{\bf Step 3.}\medskip

Now we write a lower bound for $\|(-i\nabla-bA_j^{\rm lin} )v_j\|^2_\Omega$. Notice that $\curl A_j^{\rm lin}=B(p_j)$ is constant. 
As recalled in Appendix~\ref{appendix.A},
\[\begin{aligned}
\|v_j\|^2_\Omega&=\int_{\R_+^2} |\tilde v_j|^2(1-tk(s))dtds,\\
\|(-i\nabla-bA_j^{\rm lin} )v_j\|^2_\Omega& =
\int_{\R_+^2} \Bigl( |\partial_t\tilde v_j|^2+(1-tk(s))^{-2}|(-i\partial_s^2+B(p_j)b(t-\mbox{$\frac12$}t^2k(s)))\tilde v_j|^2 \Bigr) (1-tk(s))dtds,
\end{aligned} \]
where 
\[
\tilde v_j(t,s)=e^{-ib \varphi_j(t,s)}v\circ\Phi_0(t,s)
\]
is obtained after expressing $v_j$ in the $(t,s)$ coordinates and after performing a gauge transformation.

\vspace{0.1cm}\noindent
Recall that the support of $v_j$ is contained in $\{\mathrm{dist}(x,\partial\Omega)\leq 2b^{-\rho}\}\cap\mathrm{supp\,}g_j$. Consequently, on the support of $\tilde v_j$,
\[(1-tk(s))= 1+\mathcal O(b^{-\rho}),\quad (1-tk(s))^{-2} =1+\mathcal O(b^{-\rho}),\]
and a routine application of H\"older's inequality yields,
\[ |(-i\partial_s+b(t-\mbox{$\frac12$}t^2k(s)))\tilde v_j|^2\geq (1-\epsilon)|(-i\partial_s+B(p_j) bt)\tilde v_j|^2+\mathcal O(\epsilon^{-1}b^{2-4\rho}|\tilde v_j|^2).\]
Collecting the previous estimates, we get
\[
\|(-i\nabla-bA_j^{\rm lin} )v_j\|^2_\Omega 
\geq \bigl(1-\epsilon+\mathcal O(b^{-\rho})\bigr)\|(-i\nabla-B(p_j)bA_0)\tilde v_j\|_{\R_+^2}^2+\mathcal O(\epsilon^{-1}b^{2-4\rho}\|v_j\|_\Omega^2). 
\]
where we notice that $v_j$ is supported in $\overline{\mathbb R_+^2}$. Inserting this into \eqref{eq:step2*} and using \eqref{eq:HN-sec3}, we obtain
\begin{multline}\label{eq:lb-step3}
\|(-i\nabla-bA )\chi_1 u\|^2_\Omega\geq 
(1+\mathcal O(\epsilon)+\mathcal O(b^{-\rho})) \Bigl(\inf_{j}|B(p_j)|\Bigr)^{1/2}b^{1/2}\hat \alpha\sum_{j=1}^2\|\tilde v_j\|^2_{\partial\mathbb R_+^2}+\frac{\epsilon\Theta_1b}{4}\|v\|_\Omega^2.
\end{multline}

\vspace{0.1cm}\noindent
{\bf Step 4.}\medskip

To finish the proof, we observe that
\[
\sum_j \|\tilde v_j\|^2_{\partial\mathbb R_+^2}=\int_{\partial\Omega}\Biggl(\sum_jg_j^2\Biggr)|u|^2ds(x)=\|u\|_{\partial\Omega}^2, \]
 and we choose $(\rho,\delta)$ such that  
 \[\frac14+\frac{\delta}2<\rho<\frac{\zeta-\delta}{2} \mbox{ and }0< \delta < \frac12\Bigl(\zeta-\frac12\Bigr)\,,\]
to ensure that the conditions in \eqref{eq:cond-rho1} and \eqref{eq:cond-rho2} are respected.

\vspace{0.1cm}\noindent
We insert \eqref{eq:lb-step3} into \eqref{eq:step1} and we recall that  $\epsilon=b^{-\delta}$. Eventually,  if we choose $\rho>\delta$ we obtain
 \[  \|(-i\nabla-bA)u\|^2_\Omega\geq \bigl(1+\mathcal O(b^{-\delta})\bigr)\,\Bigl(\inf_{x\in\partial\Omega}|B(x)|\Bigr)b^{1/2}\hat \alpha\|u\|_{\partial\Omega}^2.\]
 Using the variational formulation of $\lambda^{\rm DN}_1(bA,\Omega)$ achieves  the proof.
\end{proof}

\subsection{Concentration of the magnetic harmonic extension}

In the case of a magnetic field which is constant on a neighborhood of $\partial\Omega$, we would like to get a more accurate lower bound for $\lambda^{\rm DN}(bA,\Omega)$ that matches with the upper bound in Proposition~\ref{prop:ub-2D-B=1}. As an intermediate step, we need some information on the concentration of the magnetic harmonic extension of ground states.

For all $b>0$, let $f_b$ be a normalized eigenfunction (in $L^2(\partial\Omega)$) of the D-to-N eigenvalue $\lambda^{\rm DN}(bA,\Omega)$. Let us denote by $u_b$ its magnetic harmonic extension to $\Omega$, that is
\begin{equation}\label{eq:conc-he}
(-i\nabla-bA)^2u_b=0\quad\mbox{ on }\Omega,\quad u_b|_{\partial\Omega}=f_b.
\end{equation}
The weak formulation of the D-to-N map yields
\begin{equation}\label{eq:conc-wf}
    \forall\, v\in H^1(\Omega),\quad \langle (-i\nabla-bA)u_b,(-i\nabla-bA)v\rangle_\Omega-\lambda^{\rm DN}(bA,\Omega)
    \langle f_b,v|_{\partial\Omega}\rangle_{\partial\Omega}=0,
\end{equation}
where $\langle\cdot,\cdot\rangle_\Omega$ and $\langle\cdot,\cdot\rangle_{\partial\Omega}$ denote the inner product in $L^2(\Omega)$ and $L^2(\partial\Omega)$, respectively.

The next proposition states that $u_b$ decays exponentially away from the boundary. 
\begin{proposition}\label{prop:conc-dec}
Let $\Omega$ be a regular domain in $\mathbb R^2$ and $A$ be a vector potential with a magnetic field $B=\curl A$  that does not vanish on $\overline\Omega$.
    Suppose that $B$  is $C^1$  on $\overline{\Omega}$ and that $B=1$  on a neighborhood of $\partial\Omega$. Let
    \[0<\delta<m(B,\Omega):=\inf_{x\in\overline{\Omega}}|B(x)|.\]
    Then, there exists $C_\delta,b_\delta>0$ such that, for all $b\geq b_\delta$, the magnetic harmonic extension $u_b$ of $f_b$ satisfies
\[\begin{gathered}
\int_{\Omega}|u_b|^2\exp\bigl(\delta b^{1/2}\,t(x)\bigr)\,dx\leq C_\delta\int_\Omega|u_b|^2dx,\\
    \int_{\Omega}|(-i\nabla-bA)u_b|^2\exp\bigl(\delta b^{1/2}\,t(x)\bigr)\,dx\leq C_\delta b\int_\Omega|u_b|^2dx,
\end{gathered}\]
where $t(x)=\mathrm{dist}(x,\partial\Omega)$.
\end{proposition}

\begin{proof} The proof relies on the method of Agmon estimates, with due adjustments to fit the D-to-N operator.\medskip

\vspace{0.1cm}\noindent
{\bf Step 1. \it Link with a magnetic Robin Laplacian.}\medskip

To lighten the  notation, we will write $\lambda(b)$ for $\lambda^{\rm DN}(bA,\Omega)$.  Knowing from the previous sections  that $\lambda(b)=\hat \alpha b^{1/2}+o(b^{1/2})$, we have the following lower bounds \cite[Theorem~1.1 (2)]{Ka}, for all $\psi\in H^1(\Omega)$  with support in $\{B=1\}$, 
\begin{equation}\label{eq:conc-qf-lb}
\|(-i\nabla-bA)\psi\|_\Omega^2 -\lambda(b)\|\psi\|_{\partial\Omega}^2
\geq \bigl(\Theta(-\hat \alpha)b+o(b) \bigr)\|\psi\|_\Omega^2,
\end{equation}
and
\begin{equation}\label{eq:conc-qf-lb*}\|(-i\nabla-bA)\psi\|_\Omega^2 -2\lambda(b)\|\psi\|_{\partial\Omega}^2
\geq \bigl(\Theta(-2\hat \alpha)b+o(b) \bigr)\|\psi\|_\Omega^2,
\end{equation}
where $\Theta(-\hat\alpha)=0$ and $\Theta(-2\hat\alpha)<0$, by the considerations in Subsection~\ref{sec:1D-Robin}. 

\vspace{0.1cm}\noindent
Essentially, the lower bounds in \eqref{eq:conc-qf-lb} and \eqref{eq:conc-qf-lb} result from 
\[\|(-i\nabla-bA)\psi\|_\Omega^2 -\gamma\|\psi\|_{\partial\Omega}^2
\geq \bigl(\Theta(\gamma,b)+o(b) \bigr)\|\psi\|_\Omega^2,\]
the scaling relation in \eqref{eq:DG-scaling}, and the continuity of $\Theta(\cdot)$.\medskip

\vspace{0.1cm}\noindent
{\bf Step 2.}\medskip

Using \eqref{eq:conc-wf} with $v=w^2u_b$ and $w=\exp(\delta b^{1/2} t(x))$, we get by Proposition~\ref{prop:identity},
\begin{equation}\label{eq:conc-Agmon-id}
0=\|(-i\nabla -bA)wu_b\|_\Omega^2 -\lambda(b)\|u_b\|_{\partial\Omega}^2
-\delta^2 b\|wu_b\|^2_\Omega.
\end{equation}
Consider $R>1$ to be chosen sufficiently large, and a partition of unity \[\chi_1^2+\chi_2^2=1\quad\mbox{on }\Omega,\]
where $\chi_1=1$ on $\{\mathrm{dist}(x,\partial\Omega)<R b^{-1/2}\} $, $  \supp \chi_1\subset  \{\mathrm{dist}(x,\partial\Omega)< 2 R b^{-1/2}\} $ and
\[ |\nabla\chi_1| +|\nabla\chi_2|=\mathcal O(R^{-1}b^{1/2)}).\] We then have
\[\|(-i\nabla -bA)(wu_b)\|_\Omega^2=\sum_{j=1}^2\|(-i\nabla -bA)(\chi_jwu_b)\|_\Omega^2+\mathcal O(R^{-2}b)\|wu_b\|^2_\Omega,\]
with 
\[\|(-i\nabla -bA)(\chi_2 wu_b)\|_\Omega^2\geq m(B,\Omega)b\|\chi_2 wu_b\|^2_\Omega,\]
which follows from \cite[Lemma~1.4.1]{FH4}, since $\chi_2 wu_b$ is compactly supported and $B$ does not vanish in $\Omega$

Eventually, by using \eqref{eq:conc-qf-lb} with $\psi=\chi_1wu_b$, and noticing that $w|_{\partial\Omega}=1$, we get from \eqref{eq:conc-Agmon-id},
\[0\geq o(b)\|\chi_1wu_b\|_\Omega^2 +m(B,\Omega)b\|\chi_2wu_b\|_\Omega^2
-\bigl(\delta^2 b+\mathcal O(R^{-2}b)\big)\|wu_b\|_\Omega^2.\]
On the support of $\chi_1$, we know that $w=\mathcal O(1)$, so we have
\[\bigl(m(B,\Omega)-\delta^2+\mathcal O(R^{-2}) \bigr)b\|\chi_1 w u_b\|_\Omega^2\leq C b \|u_b\|_\Omega^2, \]
where the constant $C$  depends on $\delta$ and $R$. \medskip

\vspace{0.1cm}\noindent
{\bf Step 3.}\medskip

To conclude, since $0<\delta<m(B,\Omega)\leq 1$,  we choose $R$ sufficiently large so that 
\[ m(B,\Omega)-
\delta^2+\mathcal O(R^{-2})\geq \frac12(m(B,\Omega)-\delta^2)\]
and obtain the estimate
\[\|wu_b\|_\Omega^2=\mathcal O(\|u_b\|_\Omega^2).\]
Returning to \eqref{eq:conc-Agmon-id} and writing
\[\delta^2 b\|wu_b\|_\Omega^2=\frac12\|(-i\nabla -bA)(wu_b)\|_\Omega^2+\frac12\|(-i\nabla -bA)(wu_b)\|_\Omega^2-\lambda(b)\|u_b\|_{\partial\Omega}^2,\]
we get from \eqref{eq:conc-qf-lb*},
\[\delta^2 b\, \|wu_b\|_\Omega^2\geq \frac12\|(-i\nabla -bA)(wu_b)\|_\Omega^2+\bigl(\frac 12 \Theta(-2\hat \alpha)b+o(b))\|wu_b\|_\Omega^2,\]
which eventually yields
\[\|(-i\nabla -bA)(wu_b)\|_\Omega^2=\mathcal O(b\|wu_b\|^2_\Omega).\]
\end{proof}
\begin{remark}
    \label{rem:gs-dec}
The decay estimates in Proposition~\ref{prop:conc-dec} continue to hold if $u_b$ is the magnetic harmonic extension of 
 a $f_b\in H^{1/2}(\partial\Omega)$ with 
the D-to-N energy $\langle\Lambda_{bA}f_b,f_b\rangle\leq  Cb^{1/2}$, where $C\in\R_+$.     
\end{remark}
\begin{corollary}\label{corol:decay}
Given an integer $m\geq 0$, there exists $C_m,b_m>0$ such that, for all $b\geq b_m$, 
\[\begin{gathered}\int_\Omega |u_b|^2 \bigl(t(x)\bigr)^mdx\leq C_m\,  b^{-\frac{m+1}{2}},\\ \int_\Omega|(-i\nabla-bA)u_b|^2 \bigl(t(x)\bigr)^mdx\leq C_m\,  b^{-\frac{m-1}{2}}.
\end{gathered}\]
\end{corollary}
\begin{proof}
Knowing that $\|u_b\|_{\partial\Omega}=1$ and $\lambda^{\rm DN}(bA,\Omega)\leq \alpha b^{1/2}+o(b^{1/2})$,  \eqref{eq:def-var-DN} yields
\[\Theta_1b\|u_b\|^2_\Omega= \mathcal O(b^{1/2})\,.
\]
The result in the corollary   follows  from Proposition~\ref{prop:conc-dec} since $z^m\leq m!\, e^z$ for $z\geq 0$.
\end{proof}
\subsection{Two terms asymptotics}
The main idea, following what has been done in Surface Superconductivity is to use the result in \cite{HeNi1} for disks, the radius being locally chosen as the inverse of the curvature when it is positive. The starting point is in the case of the disk $\cB_R$ of radius $R$
\begin{equation}\label{eq:asymp-disk}
\lambda^{\rm DN}(bA,\cB_R)= \hat \alpha b^{1/2}  - \frac{\hat \alpha^2 +1}{3}  R^{-1}+ \mathcal O (b^{-1/2})\,.
\end{equation}
 The analysis in \cite{HeNi1} can be applied to $\cB_R^{\rm ext}$, the exterior of the disk $\cB_R$, and we get
\begin{equation}\label{eq:asymp-disk-ext}
\lambda^{\rm DN}(bA,\cB_R^{\rm ext})= \hat \alpha b^{1/2} + \frac{\hat\alpha^2 +1}{3} R^{-1}+ \mathcal O (b^{-1/2})\,.
\end{equation}
Hence we can also consider boundary points with negative curvature.  We will give another proof of \eqref{eq:asymp-disk-ext} below, which relies on a known result for a model with a Robin boundary condition.

\vspace{0.1cm}\noindent
To cover later all the cases with one notation we introduce for $R\in \mathbb R$
\[
\lambda^{\rm DN}(b,R) =\left\{
\begin{array}{ll} \lambda^{\rm DN}(bA,\cB_R)& \mbox{ if } R>0\\
 \hat \alpha b^{1/2}& \mbox{ if } R=0 \\
 \lambda^{\rm DN}(bA,\cB_{-R}^{\rm ext})& \mbox{ if } R<0\,,
 \end{array}
\right.
\]
and observe that
\begin{equation}\label{eq:asymp-const-curv}
\lambda^{\rm DN}(b,R) = \hat \alpha b^{1/2}  - \frac{\hat \alpha^2 +1}{3}  R^{-1}+ \mathcal O (b^{-1/2})\,.
\end{equation}

\vspace{0.1cm}\noindent
We will prove the following proposition.
\begin{proposition}\label{prop:lb-2D-B=1}
Let $\Omega$ be a regular domain in $\mathbb R^2$ and $A$ be a vector potential with a magnetic field $B=\curl A$ that does not vanish on $\overline\Omega$.
    Suppose that $B$  is $C^1$  on $\overline{\Omega}$ and that $B=1$  on a neighborhood of $\partial\Omega$. Then, the ground state energy of the D-to-N map $\Lambda_{b A}$ satisfies
 \begin{equation*}
 \lambda^{\rm DN}(b A,\Omega) \geq  \hat\alpha\, b^{\frac{1}{2}} -  \frac{\hat \alpha^2+1}{3}\, \max_{x\in\partial\Omega}k(x)+\mathcal O(b^{-1/6})\quad , \quad b\to+\infty\,,
 \end{equation*}
 where $k$ is the curvature of $\partial\Omega$.
\end{proposition}

\subsubsection{Warmup}

Consider $R_*>0$, $\rho\in(0,\frac12)$ and
a smooth complex-valued function $\tilde u$ on $(0,t_0)\times(\mathbb R/L\mathbb Z)$ such that 
\[\mathrm{supp}\,\tilde u\subset [s_1,s_2]\times(0,b^{-\rho}),\]
where $|s_2-s_1|<2\pi R_*<L$, $b\geq b_0$ and $R_*>b^{-\rho}$.
We introduce the energy
\[
    q_*(\tilde u)=
    \int_{\mathbb R/L\mathbb Z}\int_0^{t_0}\Bigl(|\partial_t\tilde u|^2+(1-R^{-1}_*t)^{-2}|(-i\partial_s+b(t-\mbox{$\frac12$}R^{-1}_*t^2)\tilde u| \Bigr)^2(1-R_*^{-1}t)dtds.
\]
Using \eqref{eq:asymp-disk} and reverting to polar coordinates with the change of variables $r=R_*-t$, $\theta=s-s_1$, we get
\begin{equation}\label{eq:warmup}
q_*(\tilde u)\geq  \lambda^{\rm DN}(b,{R_*}) \int_{\mathbb R/L\mathbb Z}|\tilde u(0,s)|^2ds.
\end{equation}

\vspace{0.1cm}\noindent
A similar analysis applies if $R_*\leq 0$, by using \eqref{eq:asymp-const-curv}, and we find that \eqref{eq:warmup} continues to hold in this case too. 
\subsubsection{A perturbed model}

To deal with the exterior of a  disk, or  more general exterior/interior domains, we consider an approximate model with a constant curvature $\beta$. 

Consider $\beta\in\R$, $\rho\in(\frac14,\frac12)$, and the following energy
\[q^{\rm app}_\beta(u)=\int_{\R}\int_0^{b^{-\rho}}\Bigl(|\partial_tu|^2+(1+2t\beta)
\bigl|\bigl(-i\partial_s+b(t-\mbox{$\frac12$}t^2\beta)\bigr)u|^2\Bigr)(1-t\beta)dtds, \]
where $u\in C_c^\infty\bigl(\R\times(0,b^{-\rho})\bigr)$. 
\begin{lemma}
    \label{lem:pert-model}
    There exist constants $C,b_0>0$, such that, for $b\geq b_0$ and $u\in C_c^\infty\bigl(\R\times(0,b^{-\rho})\bigr)$, we have
    \[q^{\rm app}_\beta(u)\geq \Bigl(b^{1/2}\hat\alpha-\frac{\hat\alpha^2+1}{3}\beta -Cb^{-\frac12}\Bigr)\int_{\R}|u(0,s)|^2ds. \]
    Moreover, the constants $C,b_0$ can be chosen independently of $\beta$ when it varies in a bounded interval.
\end{lemma}
\begin{proof}~\medskip

\vspace{0.1cm}\noindent
{\bf Step 1.}\medskip

Let  us introduce
\[\lambda=\inf_{\substack{u\in C_c^\infty(\R\times(0,b^{-\rho}))\\u|_{\{0\}\times\R}\not=0}}\frac{q^{\rm app}_\beta(u)}{\int_\R|u(0,s)|^2ds}.\]
Arguing as in Propositions~\ref{prop:ub-2D} and Proposition~\ref{prop:lb}, we can show that
\[\lambda=\hat\alpha b^{1/2}+o(b^{1/2}).\]

\vspace{0.1cm}\noindent
{\bf Step 2.}\medskip 

For $d\in\R$, let $$\hat\lambda(d)=\hat\alpha b^{1/2}+d\,.$$
We introduce the ground state energy
\[G(\hat\lambda(d))=\inf_{\substack{u\in C_c^\infty(\R\times(0,b^{-\rho}))\\u\not=0}}
\frac{q^{\rm app}_\beta(u)-\hat\lambda(d)\int_\R|u(0,s)|^2ds}{\int_{\R}\int_0^{b^{-\rho}}|u|^2(1-\beta t)dtds}.\]
With $\gamma=-b^{-1/2}\hat\lambda(d)$ and $\Theta(\gamma)$ as in \eqref{eq:DG-scaling},  we know from  \cite[Lemma~V.9]{Ka06} 
\[G(\hat\lambda(d))= \Bigl(\Theta(\gamma) b- \beta C_1(\gamma)b^{1/2}+\mathcal O(b^{-\rho+\frac14})\Bigr)\int_{\R}\int_0^{b^{-\rho}}|u|^2(1-\beta t)dtds ,\]
where
\[ C_1(\gamma)=-\bigl\langle \bigl((\tau-\xi(\gamma))^3+\partial_\tau\bigr)f_\gamma,f_\gamma\bigr\rangle_{L^2(\R_+)},\]
and $f_\gamma$ is a normalized ground state of $\Theta(\gamma)$.\medskip

\vspace{0.1cm}\noindent

\vspace{0.1cm}\noindent
This constant is calculated in \cite[Lemma~B.4]{FLRV},
\begin{equation}\label{eq:def-C1}
C_1(\gamma)=\frac13\bigl(1-\gamma\xi(\gamma)\bigr)|u_\gamma(0)|^2.
\end{equation}
Knowing that $\gamma=-\hat\alpha -d b^{-1/2}$ and that $\Theta(\cdot)$, $\xi(\cdot)$, and $u_{\cdot}$ are smooth,  we get by Taylor's expansion at $\hat \alpha$
\[\begin{gathered}
\Theta(\gamma)=-d\Theta'(-\hat\alpha)b^{-1/2}+\mathcal O(d^2b^{-1}),\\ 
\xi(\gamma)=\hat\alpha-\frac{d\Theta'(-\hat\alpha)}{2\hat\alpha}b^{-1/2}+\mathcal O(d^2b^{-1}),\\
C_1(\gamma)=
\frac13(1+\hat\alpha^2)|u_{-\hat\alpha}(0)|^2+\mathcal O(db^{-1/2}).\end{gathered}\]
We can choose $\zeta>0$,
\[ d=d_*=-\frac{\beta}{3}(1+\hat\alpha^2)\frac{|u_{-\hat\alpha}(0)|^2}{\Theta'(-\hat\alpha)}-\zeta b^{-1/2},\] 
and $\hat\lambda_*=\hat\lambda(d_*)$ such that 
\[G^{\rm app}(\hat\lambda_*):=\Theta(\gamma) b- \beta C_1(\gamma)b^{1/2}+\mathcal O(b^{-\rho+\frac14})\geq 0.\]

\vspace{0.1cm}\noindent
{\bf Step 3.}\medskip

Returning to $\lambda$ introduced in Step~1, and by writing
\[ q^{\rm app}(u)=q^{\rm app}(u)-\hat\lambda_*\int_{\R}|u(0,s)|^2ds+\hat\lambda_*\int_{\R}|u(0,s)|^2ds
\]
we get
\[ q^{\rm app}(u)\geq G^{\rm app}(\hat\lambda_*)\int_0^{b^{-\rho}}|u|^2(1-\beta t)dtds+\hat\lambda_*\int_{\R}|u(0,s)|^2ds\geq \hat\lambda_*\int_{\R}|u(0,s)|^2ds, \]
and consequently,
\[\lambda\geq\hat\lambda_*= \hat\alpha b^{1/2}+d_*. \]
To finish the proof, notice that, by \cite[Proposition~II.5]{Ka06},  $$ \Theta'(-\hat\alpha)=|u_{-\hat\alpha}(0)|^2\,.$$ 
 Hence 
\[d_*=-\frac{\beta}{3}(1+\hat\alpha^2)-\zeta b^{-1/2}\,.
\]
\end{proof}
\subsubsection{ Reduction to a tubular domain}

To simplify the presentation, suppose that $\Omega$ is simply connected. Let $f_b$ be a normalized eigenfunction of the D-to-N operator, and denote by $u=u_b$ the magnetic harmonic extension of $f_b$. Consider 
a partition of unity on $\Omega$, $\chi_1^2+\chi_2^2=1$, with $\mathrm{supp}\,\chi_1\subset \{t(x)<b^{-\rho}\}$, where $t(x)=\mathrm{dist}(x,\partial\Omega)$. Then, thanks to Proposition~\ref{prop:conc-dec} and Corollary~\ref{corol:decay},
\begin{equation}\label{eq:disk}
\begin{aligned}
\lambda^{\rm DN}(bA,\Omega)&=\|(-i\nabla-bA)u\|^2_\Omega\\
&=\|(-i\nabla-bA)\chi_1u\|^2_\Omega+\mathcal O(b^{-\infty}).
\end{aligned}
\end{equation}
The support of $\chi_1u$ is contained in $\Omega_0=\{t(x)<t_0\}$, which is diffeomorphic to $\mathcal S:=(0,t_0)\times (\mathbb R/L\mathbb Z)$. 

\subsubsection{Reduction to  constant curvature}

We work with the parallel coordinates $(t,s)=\Phi_0^{-1}(x)$ and argue as in Step~2 in the proof of Proposition~\ref{prop:lb}. In fact, we introduce a partition of unity of $\mathbb R/L\mathbb Z$
\[\sum_{j}g_j^2=1,\quad \sum_{j}|\nabla g_j|=\mathcal O(b^{\delta}),\quad\mathrm{supp}\,g_j\subset[s_{j}-b^{-\delta},s_{j}+b^{-\delta}],\]
where $0<\delta<\frac12$ is to be chosen later on. We write
\begin{equation}\label{eq:curv-dec}
\|(-i\nabla-bA)\chi_1u\|^2_\Omega=\sum_{j}q(\tilde u_j)+\mathcal O(b^{2\delta}\|u\|_\Omega^2),\end{equation}
where
\[ q(\tilde u_j)=\int_{\mathcal S}\Bigl( |\partial_t\tilde u_j|^2+(1-tk(s))^{-2}|(-i\partial_s-b\,a(t,s))\tilde u_j|^2
\Bigr)(1-tk(s))dsdt, \]
and
\[a(s,t)=-t+\mbox{$\frac12$}t^2k(s),\quad \tilde u_j=g_j(s)e^{i\varphi_j(s,t)} (\chi_1 u\circ\Phi_0^{-1})(s,t).\]
Writing
\[a_j(s,t)=-t+\mbox{$\frac12$}R_j^{-1}t^2,\quad 
R_j^{-1}=\begin{cases}k_j&\mbox{if }k_j\not=0\\
0&\mbox{if }k_j=0\end{cases},\quad k_j=k(s_j),   \]
we have
\[
\begin{gathered}a(s,t)=a_j(s,t)+\mathcal O(b^{-\delta}t^2),\\
1-tk(s)=1-R_j^{-1}+\mathcal O(b^{-\delta}t),\\ 
(1-tk(s))^{-2}=(1-R_j^{-1}t)^{-2}+\mathcal O(b^{-\delta}t),\end{gathered}\]
and
\[q(\tilde u_j)=q_j(\tilde u_j)+\mathcal O(E_{j,1})+\mathcal O(E_{j,2})+\mathcal O(E_{j,3})\]
where
\[\begin{gathered}
   q_j(\tilde u_j)=\int_{\mathcal S}\Bigl(|\partial_t\tilde u|^2+(1-R_j^{-1}t)^{-2}|(-i\partial_s+b(t-\mbox{$\frac12$}R^{-1}_jt^2)\tilde u|^2 \Bigr)(1-R_j^{-1}t)dtds\\
   E_{j,1}=b^{-\delta}\int \bigl(|\partial_t\tilde u_j|^2+|(-i\partial_s-b\,a)\tilde u_j|^2\bigr)t\,dtds,\\
   E_{j,2}=b^{2-2\delta}\int |\tilde u_j|^2 t^4dtds,\\
   E_{j,3}=b^{1-\delta}\|t^2\tilde u_j\|_{2}\|(-i\partial_s-b\,a)\tilde u_j\|_2.
\end{gathered}\]
Using Corollary~\ref{corol:decay}, we write
\[ \sum_{j}E_{j,1}=\mathcal O(b^{-\delta}),\quad \sum_{j}E_{j,2}=\mathcal O(b^{-\frac12-2\delta}),\quad \sum_jE_{j,3}=\mathcal O(b^{-\delta}),\]
and by using \eqref{eq:warmup},  we get
\[ q(\tilde u_j)\geq \lambda^{\rm DN}(b,R_j)\int_{\R/L\Z}|\tilde u_j(0,s)|^2ds+\mathcal O(b^{-\delta})\quad\mbox{if }k_j>0,\]
and a similar estimate when $k_j\leq 0$.
Inserting this into \eqref{eq:curv-dec} and noticing that,
\[ \begin{gathered}
    \sum_j\int_{\R/L\Z}|\tilde u_j(0,s)|^2ds=\int_{\partial\Omega}|u|^2ds=1,\\
b^{2\delta}\|u\|_\Omega^2=\mathcal O(b^{-\frac12+2\delta})\mbox{ by Corollary~\ref{corol:decay}},
\end{gathered}\]
we get
\[ \|(-i\nabla-bA)\chi_1u\|^2\geq \inf_{j}\lambda^{\rm DN}(b,R_j)+\mathcal O(b^{-\delta})+\mathcal O(b^{2\delta-\frac12}).\]
By invoking \eqref{eq:asymp-const-curv} and using that
\[R_j^{-1}=k(s_j)\leq \max_{x\in\partial\Omega}k_x,\]
we get
\[\inf_{j}\lambda^{\rm DN}(b,R_j)\geq \hat \alpha b^{\frac12}-  \frac{\hat \alpha^2+ 1}{3} \max_{x\in\partial\Omega}k_x+\mathcal O(b^{-1/2}).\]

\vspace{0.1cm}\noindent
Optimizing over $\delta$, we choose $\delta=\frac16$, and we get from \eqref{eq:disk}
\[ \lambda^{\rm DN}(bA,\Omega)\geq \hat \alpha b^{\frac12}-  \frac{\hat \alpha^2+ 1}{3}\max_{x\in\partial\Omega}k_x+\mathcal O(b^{-\frac16}).\]

\section{Splitting of eigenvalues in 2D}\label{sec:splitting}
\begin{assumption}\label{ass:max-curv}
    Suppose that $\Omega$ is a regular domain such that the curvature $k$ has a unique maximum of curvilinear coordinate $0$, and that $k_2:=-k''(0)>0$.
\end{assumption}

\begin{theorem}\label{thm:splitting-2D}
    Suppose that Assumption~\ref{ass:max-curv} holds.  Let $A$ be a vector field on $\Omega$ with constant magnetic field $\curl A=1$. Then, for every $j\in\N$,  the $j$'th eigenvalue of the D-to-N map satisfies
    \[ \lambda_j(bA,\Omega)=\hat\alpha b^{\frac12}-\frac{\hat\alpha^2+1}{3}\max_{x\in\partial\Omega}k(x)+(2j-1)c_*b^{-1/4}+o(b^{-1/4}) \quad\mbox{ as }b\to+\infty,\]
where $c_*=c_*(\Omega)>0$ is a  constant.
\end{theorem}

\subsection{Robin eigenvalues}

We will establish a link between the D-to-N eigenvalues, and the eigenvalues of the magnetic Laplacian with Robin boundary condition. For $j\in\N$ and $\gamma\in\R$, we introduce the eigenvalues
\[\mu_j^{\rm R}(\gamma,b)=\min_{\mathrm{dim}(M)=j}\Bigl(\max_{\|u\|_\Omega=1}
\|(-i\nabla-bA)u\|^2_\Omega+ b^{1/2}\gamma\|u\|_{\partial\Omega}^2 \Bigr).\]
\begin{proposition}\label{prop:zero-R}
For all $j\in\N$ and $b>0$, the equation
\[\mu_j^{\rm R}(\gamma,b)=0\]
has a unique solution $\gamma_j(b)$, and $\gamma_j(b)<0$.  Furthermore, 
\[\lambda_1^{\rm DN}(bA,\Omega)=- b^{1/2}\gamma_1(b), \]
and if for all $j\leq N$,  the eigenvalue $\mu_j(\gamma_j(b),b)$ is simple,  then
\[\lambda_j^{\rm DN}(bA,\Omega)=- b^{1/2}\gamma_j(b)\quad \mbox{ for }2\leq j\leq N.\]
\end{proposition}
\begin{proof}
Suppose that $j$ and $b$ are fixed. The function $\gamma\mapsto \mu_j(\gamma,b)$ is increasing and we know that $\lambda_j(0,b)>0$. By the min-max principle,
\[ \mu_j^{\rm R}(\gamma,b)\leq 2\mu_j^{\rm R}(\gamma,0)+2b^2\|A\|_\infty^2,\]
and we have the spectral asymptotics  \cite{DK}
\[\mu_j^{\rm R}(\gamma,0)=-\gamma^2+o(\gamma^2)\quad(\gamma\to-\infty).\]
This proves that $\gamma\mapsto\mu_j^{\rm R}(\gamma,b)$ has a unique zero $\gamma_j(b)$, and $\gamma_j(b)<0$. Furthermore, by monotonicity, we get $\gamma_{j+1}(b)\leq \gamma_j(b)$. 

Let  $\psi_j$ be an eigenfunction of $\mu_j(\gamma_j(b),b)$, $g_j=\psi_j|_{\partial\Omega}$ is an eigenfunction of the D-to-N operator with eigenvalue $-b^{1/2}\gamma_j(b)$. Hence, $(-b^{1/2}\gamma_j(b))$ is a non-decreasing sequence of eigenvalues of the D-to-N operator $\Lambda_{bA}$.

Consider an orthonormal basis  $\{f_j\colon j\in\N\}$ in $L^2(\partial\Omega)$ such that every $f_j$ is an eigenfunction of the D-to-N  operator, corresponding to the $j$'th eigenvalue $\lambda_j^{\rm DN}(bA,\Omega)$. If $u_j$ is the magnetic harmonic extension of $f_j$, then it is a zero mode for the magnetic Laplacian $(-i\nabla-bA)^2$  with Robin boundary condition
\[ \vec{\nu}\cdot(-i\nabla-bA)u_j=\lambda_j^{\rm DN}(bA,\Omega)u_j\quad \mbox{ on }\partial\Omega.\]

Since $\{u_j\colon j\in\N\}$ is linearly independent, we get that $\lambda^{\rm DN}_j(bA,\Omega)=-b^{1/2}\gamma_j(b)$, for all $j\leq N$, provided that 
the eigenvalues $\mu_1^{\rm R}(\gamma_1(b),b),\cdots, \mu_N^{\rm R}(\gamma_N(b),b)$ are simple.
\end{proof}
\subsection{Proof of Theorem~\ref{thm:splitting-2D}}

Suppose that $j\in\N$ and $\gamma\in\R$ are fixed. If $\Omega$ satisfies Assumption~\ref{ass:max-curv}, and if the magnetic field $B=\curl A$ is equal to $1$, it follows from \cite{FLRV}  that 
\begin{equation}\label{eq:asym-Robin-F}
\mu_j^{\rm R}(\gamma,b)=b\Theta(\gamma)-b^{1/2}
C_1(\gamma)\max_{x\in\partial\Omega}k(x)+(2j-1)b^{1/4}C_2(\gamma)+o(b^{1/4})\quad(b\to+\infty),
\end{equation}
locally uniformly with respect to $\gamma$. The coefficients in \eqref{eq:asym-Robin-F} are given as follows:
\begin{itemize}
    \item $\Theta(\gamma)$ is introduced in \eqref{eq:1D-Robin*};
    \item $C_1(\gamma)$ is introduced in \eqref{eq:def-C1};
    \item $C_2(\gamma)=\frac12\sqrt{k_2C_1(\gamma)\partial_\xi^2\mu(\gamma,\xi(\gamma))}$, $\mu(\gamma,\xi)$ is introduced in \eqref{eq:1D-Robin*}, and $\xi(\gamma)$ is introduced in \eqref{eq:th-xi}.
\end{itemize}
Thanks to  Propositions~\ref{prop:ub-2D-B=1}, \ref{prop:lb-2D-B=1} and \ref{prop:zero-R},  the unique zero of $\mu_j^{\rm R}(\gamma,b)=0$ satisfies
\[\gamma_j(b)=-b^{-1/2}\lambda_j^{\rm DN}(bA,\Omega)=-\hat\alpha +\frac{\hat\alpha^2+1}{3}k_* b^{-1/2}+o(b^{-1/2}),\]
where $k_*=\max_{x\in\partial\Omega}k(x)$. Furthermore,
\[ \Theta(-\hat\alpha)=0,\quad C_1(-\hat\alpha)=\frac{\hat\alpha^2+1}{3}\Theta'(-\hat\alpha).\]
We write by Taylor's formula at $-\hat\alpha$, 
\[\begin{aligned}
    \Theta(\gamma_j(b))&=(\gamma_j(b)+\hat\alpha)\Theta'(-\hat\alpha)+\mathcal O(b^{-1}),\\
C_1(\gamma_j(b))&=C_1(-\hat\alpha)+\mathcal O(b^{-1/2})=\frac{\hat\alpha^2+1}{3}\Theta'(-\hat\alpha)+\mathcal O(b^{-1/2}),\\
C_2(\gamma_j(b))&=C_2(-\hat\alpha)+o(1)\qquad (b\to+\infty).
\end{aligned}\]
Inserting these formulas into \eqref{eq:asym-Robin-F} and noting that $\mu_j^{\rm R}(\gamma_j(b),b)=0$, we get 
\[\gamma_j(b)=-\hat\alpha+\frac{\hat\alpha^2+1}{3}k_*b^{-1/2}-(2j-1)\frac{C_2(-\hat\alpha)}{\Theta'(-\hat\alpha)}b^{-3/4}+o(b^{-3/4}).\]
Thanks to Proposition~\ref{prop:zero-R}, we finish the proof of Theorem~\ref{thm:splitting-2D}, where the constant $c_*$ is 
\[c_*=\frac{C_2(-\hat\alpha)}{\Theta'(-\hat\alpha)}.\]

\section{The case of dimension 3}\label{sec:3D}

\subsection{Introduction}
 
\vspace{0.1cm}\noindent
In the unbounded case, the definition of the magnetic D-to-N map  is not quite as simple as in the case of bounded domains. For compactly supported magnetic fields,  the  D-to-N operator in the half-space $\R^3_+$  was well defined  in (\cite{Li}, Appendix B) using the Lax-Phillips method. For such  compactly supported magnetic fields, the solvability of the direct problem in an infinite slab $\Sigma$  was also studied  in \cite{KrUh}. We recall that an infinite slab is defined as 
\begin{equation}
	\Sigma = \{ x= (x', x_n) \in \R^n \ :\ x'= (x_1, ...., x_{n-1}) \in \R^{n-1}\ , \ 0<x_n<L\} \ , \  n \geq 3\,.
\end{equation}
At last, for non-compactly supported electromagnetic fields, the  D-to-N map on an unbounded open set $\Omega \subset \R^3$ corresponding to a closed waveguide was studied in \cite{Ki}. Here by closed waveguide, we mean that there exists a $C^2$ bounded open simply connected set $\omega \subset \R^2$ such that $\Omega \subset \omega \times \R$.

On the other hand, we do not really need to introduce the D-to-N operator in $\mathbb R^3_+$ but only use the corresponding  ground state energy given by the variational approach and this does not involve explicitly a D-to-N operator. So we choose
 to avoid to refer to this operator in the proofs and will come back to this question which is interesting in itself in the last subsection.

\subsection{The case of the half-space}

Following what was done in Surface Superconductivity (see Lu-Pan \cite{LuPa} and Helffer-Morame \cite{HeMo3}), we have to look at the non-homogeneous Dirichlet  problem in $\mathbb R^3_+=\{(x_1,x_2,x_3) \in \R^3\ ,\ x_1 >0\}$ for the family, parametrized by the angle $(\frac \pi 2 -\vartheta$) of the magnetic field $\vec{H}$  (considered as a vector in $\mathbb R^3$) with the  normal vector $(1,0,0)$ at a point of $\partial \mathbb R^3_+$.

\vspace{0.1cm}\noindent
As in (\cite{HeNi1}, Section 3), after scaling, we can always assume that the magnetic field $b=1$. Thus, we consider the Dirichlet realization of the magnetic Laplacian  in $\mathbb R^3_+$ :
\begin{equation}\label{eq:1.4aa}
H^{\rm Dir} (\vartheta):= D_{x_1}^2 + D_{x_2}^2 + (D_{x_3} + \cos \vartheta x_1-\sin \vartheta x_2)^2\ \ ,\ \vartheta \in [0, \frac{\pi}{2}]\,,
\end{equation}
 and consider (we implicitly assume the condition that the denominator is not zero)
 
\begin{equation}\label{eq:6.2}
\lambda^{DN} (\vartheta):= \inf_{\phi\in C_0^\infty(\overline{\mathbb R^3_+})} \frac{||\nabla_{A_ \vartheta}\phi||^2}{ \int |\phi(0,x_2,x_3)|^2 dx_2 dx_3}\,,
\end{equation}
	where 
\[
A=A_\vartheta=(0,0, -\cos \vartheta x_1+ \sin \vartheta x_2)\,.
\]

\vspace{0.1cm}\noindent
The main result of this subsection is the following :

\begin{proposition}\label{basduspectre}
One has :
\begin{equation}
\inf_{\vartheta \in [0, \frac{\pi}{2}]} \lambda^{DN} (\vartheta) =\hat \alpha\:=\alpha / \sqrt{2}\,,
\end{equation}
and this infimum is uniquely realized at $\vartheta=0$.
\end{proposition}

\begin{remark}
	It follows from Proposition \ref{basduspectre} that   $\lambda^{DN} (\vartheta)$ is minimal  when the magnetic field $\vec{H}$ is parallel to the hyperplane $x_1 = 0$.
\end{remark}

\vspace{0.1cm}\noindent
To prove this proposition, we first make  a partial Fourier transform in the $x_3$-variable and we can reduce the computation to a $(\vartheta, \tau)$-family of operators on $\mathbb R^2_+=\{x_1>0\}$:
\begin{equation}
H^{\rm Dir}(\vartheta, \tau) := D_{x_1}^2 + D_{x_2}^2 + (\tau + \cos \vartheta x_1-\sin \vartheta x_2)^2\ ,\ \tau \in \R\,.
\end{equation}
We now introduce
\[
\lambda^{\rm DN} (\vartheta,\tau):= \inf_{\phi\in C_0^\infty(\overline{\mathbb R^2_+})}\frac{ ||\nabla_{\check A}\phi||^2}{ \int |\phi(0,x_2)|^2 dx_2 }\,,
\]
where 
\[
\check A=\check A_{\vartheta,\tau} =(0,0, \tau +\cos \vartheta x_1- \sin \vartheta x_2)\,.
\]
Our first result is
\begin{lemma}
\[
\lambda^{\rm DN} (\vartheta):= \inf_{\tau \in \mathbb R} \lambda^{\rm DN} (\vartheta,\tau)\,.
\]
\end{lemma}
\begin{proof} For the upper-bound, we can use sequences of the form $\chi_n(t) e^{it\tau} \theta(x_1,x_2)$. For the lower bound, we use the partial Fourier transform.
\end{proof}
The next point is to observe that:
\begin{lemma}\label{tau0}
 If $\vartheta \in (0,\frac \pi 2]$,
\begin{equation}
\lambda^{\rm DN}(\vartheta,\tau)= \lambda^{DN}(\vartheta,0)\,.
\end{equation}
\end{lemma}
\begin{proof}
This is evident through a translation in the tangential $x_2$ variable.
\end{proof}
We will also need
\begin{lemma}\label{lemma6.4}
 If $\vartheta \in (0,\frac \pi 2)$, there exists $\phi_{\vartheta} \in H^1_{\tilde A}(\mathbb R^2_+)$ such that
 \[
 \int \phi_{\vartheta}(0,x_2)^2 dx_2=1
 \]
 and 
 \[
 \lambda^{\rm DN}(\vartheta,0) = ||\nabla_{\tilde A} \phi_{\vartheta} ||^2
\]
where
\[
\tilde A= \check A_{\vartheta,0}\,.
\]
\end{lemma}
\begin{proof}
This is a consequence of the compact injection of $H^1_{\tilde A}(\mathbb R^2_+)$ in $L^2(\mathbb R^2_+)$\,.
\end{proof}

\vspace{0.1cm}\noindent
Hence, the parameter $\tau$ is relevant only in the case when $\vartheta=0$, and this is the object of the following lemma :

\begin{lemma}
If $\vartheta=0$,
\begin{equation}
\lambda^{\rm DN}(0,\tau)= \hat \alpha\,.
\end{equation}
\end{lemma}
\begin{proof}
In this case, the Laplacian becomes
\begin{equation}
H(0, \tau) = D_{x_1}^2 + D_{x_2}^2 + (\tau + x_1)^2\,.
\end{equation}
After  a Fourier transform in the $x_2$-variable, we get a new family of magnetic Laplacians :
\begin{equation}
\tilde{H}( \xi_2,\tau ) = D_{x_1}^2 + \xi_2^2 + (\tau + x_1)^2\ \ ,\ \xi_2 \in \R\,.
\end{equation}
We now have to analyze the family of  the associated energies $\lambda^{DN}(\xi_2,\tau)$ depending of two parameters $(\xi_2, \tau)$. Using the variational  characterization of the ground state, it is clear that the infimum is obtained for $\xi_2=0$, and this latest case was analyzed in (\cite{HeNi1}, Proposition 3.1). Thus, the proof is complete.
\end{proof}

\vspace{0.1cm}\noindent
Now, let us study the case $\vartheta = \frac{\pi}{2}$. One gets :

\begin{lemma}
If $\vartheta=\frac \pi 2$,
\begin{equation}
  \lambda^{\rm DN}(\frac{\pi}{2},\tau) =1\,. 
\end{equation}
\end{lemma}

\begin{proof}
By Lemma \ref{tau0}, we can always assume that $\tau=0$. So, in this case, we have :
\begin{equation}
H\bigl(\frac{\pi}{2}, 0\bigr) = D_{x_1}^2 + D_{x_2}^2 + x_2^2 \,.
\end{equation}
Thanks to  a decomposition using the Hermite functions basis in the $x_2$ variable, we have to look at the family of Hamiltonians in $\mathbb R^+$:
\begin{equation}
\hat{H}(k) = D_{x_1}^2 + (2k+1) \ ,\ k \in \N\,.
\end{equation}
If $\lambda(k)$ denotes the lowest eigenvalue of D-to-N map associated with $\hat{H}(k)$, one easily gets $\lambda(k) =  \sqrt{2k+1}$. Thus, 
\begin{equation}
\lambda^{\rm DN}\bigl(\frac{\pi}{2},\tau\bigr)=\lambda^{DN}(\frac{\pi}{2},0)=\inf_{k\geq 0} \sqrt{2k+1} =1\,.
\end{equation}
\end{proof}

\vspace{0.1cm}\noindent
{\it End of the proof of Proposition \ref{basduspectre} :}

\vspace{0.2cm}\noindent
Let us set $\Omega =\{ (x_1,x_2)\in \R^2, x_1>0 \}$.  In the following, we shall show that for any $\vartheta \in (0, \frac{\pi}{2}$), there exists a suitable $g(\vartheta) > \hat{\alpha}$ such that, for any  $u \in C^{\infty}(\overline{\Omega})$, 
\begin{equation}\label{maj}
\int_\Omega|\nabla  u| ^2 +  (x_1 \cos \vartheta-x_2 \sin \vartheta)^2 |u|^2 \ dx_1 dx_2 \geq g(\vartheta)
\int_{\partial \Omega} |u(0,x_2)|^2 \ dx_2 \,.
\end{equation}
To  this end, we follow a similar strategy  to  (\cite{HeMo3}, Subsection 3.4). We introduce an interpolation parameter $\rho \in [0,1]$, and we write the integrand in the (LHS) of (\ref{maj}) as 
\begin{multline*}
	|\nabla u| ^2  + (x_1 \cos \vartheta-x_2 \sin \vartheta)^2 |u|^2  =  \rho^2\left(   |\partial_{x_1} u| ^2   + (x_1 \cos \vartheta-x_2 \sin \vartheta)^2 |u|^2 \right) \\ 
	+ (1-\rho^2) |\partial_{x_1} u| ^2  + |\partial_{x_2} u| ^2   + (1-\rho^2) (x_1 \cos \vartheta-x_2 \sin \vartheta)^2 |u|^2 \,.
\end{multline*}

\vspace{0.1cm}\noindent
First, thanks to (\cite{HeNi1}, Proposition 3.1), one immediately gets :
\begin{equation}
	\int_0^{+\infty} \left( |\partial_{x_1} u| ^2  + (x_1 \cos \vartheta-x_2 \sin \vartheta)^2 |u|^2 \right)\ dx_1 \ \geq \ \hat \alpha  \sqrt{\cos \vartheta }\ |u(0,x_2|^2\,.
\end{equation}
Integrating with respect to the $x_2$ variable, we obtain : 
\begin{equation}\label{eq:c2}
\rho^2 \int_\Omega \left( |\partial_{x_1}  u| ^2 +  (x_1 \cos \vartheta-x_2 \sin \vartheta)^2 |u|^2 \right) \ dx_1 dx_2 \ \geq \ \rho ^2 \hat{\alpha} \sqrt{\cos \vartheta }\
\int_{\partial \Omega} |u(0,x_2)|^2 \ dx_2 \,.
\end{equation}
Secondly, we observe that :
\begin{equation}\label{variablex2}
\int_\R \left( |\partial_{x_2} u| ^2  + (1-\rho^2) (x_1 \cos \vartheta-x_2 \sin \vartheta)^2 |u|^2 \right)\ dx_2 \ \geq \ \sqrt{1-\rho^2} \sin \vartheta \int_\R  |u|^2 \ dx_2\,.
\end{equation}	
This last inequality is a consequence of the lower bound for the harmonic oscillator, (in the $x_2$ variable). 	On one hand, integrating (\ref{variablex2}) with respect to the $x_1$ variable, one easily gets :
\begin{eqnarray*}
\int_\Omega \left( (1-\rho^2) |\partial_{x_1} u| ^2 + |\partial_{x_2} u| ^2  + (1-\rho^2) (x_1 \cos \vartheta-x_2 \sin \vartheta)^2 |u|^2\right) &  dx_1 dx_2 \\ \geq  (1-\rho^2)  \left( \int_\Omega  |\partial_{x_1} u| ^2 + \right. 
 \left. (1-\rho^2)^{-\half} \sin \vartheta\  |u|^2 \ dx_1 dx_2 \right)\,.
\end{eqnarray*}
On the other hand, using the following lower bound for the D-to-N map  associated with the Hamiltonian $-\partial_{x_1}^2 + (1-\rho^2)^{-\half} \sin \vartheta $ on the interval $(0,+\infty)$,
\begin{equation}
\int_0^{+\infty}  |\partial_{x_1} u| ^2 +  (1-\rho^2)^{-\half} \sin \vartheta\  |u|^2 \ dx_1  \geq (1-\rho^2)^{-\frac{1}{4}	} \sqrt{\sin \vartheta} \  |u(0,x_2)|^2 \,,
\end{equation}		
and integrating again over the variable $x_2$, we finally get 
\begin{equation*}
\begin{array}{l}
	\int_\Omega \left( (1-\rho^2) |\partial_{x_1} u| ^2 + |\partial_{x_2} u| ^2  + (1-\rho^2) (x_1 \cos \vartheta-x_2 \sin \vartheta)^2 |u|^2\right) \ dx_1 dx_2 \\ \qquad \qquad  \geq  (1-\rho^2)^{\frac{3}{4}}  \sqrt{\sin \vartheta} \int_{\partial \Omega} |u(0,x_2)|^2 \ dx_2 \,.
\end{array}
\end{equation*}
As a conclusion, we have obtained :
\begin{equation}
	\lambda^{\rm DN}(\vartheta)  \geq  \left( \rho ^2 \hat{\alpha} \sqrt{\cos \vartheta } + (1-\rho^2)^{\frac{3}{4}}  \sqrt{\sin \vartheta} \right)\,.
\end{equation}
In particular, choosing $\rho = \cos \vartheta$, we get
\begin{equation}
 	\lambda^{\rm DN}(\vartheta)  \geq  g(\vartheta):=\hat{\alpha}	(\cos \vartheta)^{\frac{5}{2}} + (\sin \vartheta)^2 > \hat{\alpha} \ ,\ \forall \vartheta \in (0, \frac{\pi}{2}] \,.
\end{equation}
For the last inequality, we can observe that with $X=\cos \vartheta\in ]0,1]$, we have always:
$$
\hat \alpha X^{5/2} + 1 - X^2 > \hat \alpha\,.
$$
This concludes the proof of Proposition \ref{basduspectre}. \hfill $\Box$

\subsection{Lower bounds in general domains}
We have
\begin{proposition}
 	Let $\Omega$ be a regular bounded domain in $\mathbb R^3$ and $A$ be a magnetic potential with constant magnetic field with norm $1$, then the ground state energy of the D-to-N map $\Lambda^{\rm DN}_{bA}$ satisfies
 \begin{equation}
 \liminf_{b\rightarrow +\infty} b^{-1/2} \; \lambda^{\rm DN}(b A,\Omega) \geq \hat \alpha.
 \end{equation}
\end{proposition}
More generally, we have
\begin{proposition}\label{prop6.8}
	Let $\Omega$ be a regular bounded domain in $\mathbb R^3$, $A$ be a magnetic potential with non vanishing $C^\infty$ magnetic field in $\overline{\Omega}$, then, with the notation of 
	Theorem \ref{thm:main-3D}, the ground state energy of the D-to-N map $\Lambda^{DN}_{bA}$ satisfies
 \begin{equation}
 \liminf_{b\rightarrow +\infty} b^{-1/2} \; \lambda^{\rm DN}(b A,\Omega) \geq   \inf_{x\in \partial \Omega} \Big(\lambda^{DN}(\vartheta(x))|B(x)|^{\frac 12}\Big)\,.
 \end{equation}
(see  also \cite{LuPa} or \cite{HeMo3} and  \eqref{orth} in Appendix B)
\end{proposition}
\begin{proof}
We follow the proof of Proposition \ref{conj4} given in the case of dimension $2$. For step 1, we have to replace the $2D$ lower bound of the Neumann problem by the $(3D)$-statement proven by Lu-Pan \cite{LuPa}
\begin{equation}
 \liminf_{b\rightarrow +\infty} b^{-1} \lambda^{\rm Ne}(bA,\Omega) \geq \min \Big(\inf_{x\in \overline{\Omega} }   |B(x)|, \inf_{x\in \partial \Omega} \sigma(\vartheta(x)) |B(x)|\Big)\,,
\end{equation}
or a more accurate version with remainder.\\
For step 2,  the equivalent of Proposition \ref{prop:gauge-gen} is given  in Lemma 5.4 from \cite{LuPa}.\\
Finally, we can implement the constant magnetic field results obtained in the previous subsection.
\end{proof}
\subsection{Upper bounds in general domains}
As observed in \cite{Ray2} the proof for the Neumann problem is only sketched in \cite{LuPa0}. On the other hand, we only state the following version, with the notation of Proposition \ref{prop6.8}
\begin{proposition}
	Let $\Omega$ be a regular bounded domain in $\mathbb R^3$, $A$ be a magnetic potential with non vanishing $C^\infty$ magnetic field in $\overline{\Omega}$, then the ground state energy of the D-to-N map $\Lambda^{DN}_{bA}$ satisfies
 \begin{equation}\label{eq:7.2}
 \limsup_{b\rightarrow +\infty} b^{-1/2} \; \lambda^{\rm DN}(b A,\Omega) \leq   \inf_{x\in \partial \Omega} \Big(\lambda^{DN}(\vartheta(x))|B(x)|^{\frac 12}\Big)\,.
 \end{equation}
 \end{proposition}
 \begin{proof}
 We distinguish three cases depending on the value of $\vartheta(x)$ where $\inf_{x\in \partial \Omega} \Big(\lambda^{DN}(\vartheta(x))|B(x)|^{\frac 12}\Big)$ is attained.\\
 {\bf Case 1}. We assume that there exists $p \in \partial \Omega$ such that  $\vartheta (p) \in (0,\frac \pi 2)$ and 
 \begin{equation} \lambda^{\rm DN}(\vartheta(p))|B(p)|^{\frac 12}= \inf_{x\in \partial \Omega} \Big(\lambda^{\rm DN}(\vartheta(x))|B(x)|^{\frac 12}\Big)\,.
 \end{equation}
 In this case, after using Lemma 3.4 in \cite{LuPa} (which extends proposition \ref{prop:gauge-gen} to the $(3D)$-case), we can take in the new system of coordinates centered at $p$ such that $\Omega$ is locally defined by $\{x_1 >0\}$, the quasimode
 \[
 u(x_1,x_2,x_3) = b^{\frac \rho 2}\chi (b^{\rho} x_3) \chi(b^{\rho} x_1) \chi(b^{\rho} x_2) \phi_{\vartheta(p)} ((bB(p)^{1/2}x_1,(bB(p))^{1/2}x_2)\,,
 \]
 where $\phi_{\vartheta}$ is defined in Lemma \ref{lemma6.4} and $\rho\in (0,\frac 12)$.\\
 {\bf Case 2}. We assume that $\vartheta(p)=\frac{\pi}{2}$ in \eqref{eq:7.2}.\\
 We take as quasi-mode
 \[
 u(x_1,x_2,x_3) = b^{\frac \rho 2}\chi (b^{\rho} x_3) \chi(b^{\rho} x_1) \chi(b^{\rho} x_2) (bB(p))^{1/2}\exp\Bigl( - \frac 12 (bB(p)) x_2^2 \Bigr) \exp\Bigl( -(bB(p))^{1/2} x_1\Bigr)\,.
 \]
 {\bf Case 3}. We assume that $\vartheta(p)=0$ in \eqref{eq:7.2}.\\
 Notice that this is always the case when the magnetic field is constant.\\ We are essentially like in the $(2D)$ case and take
 \[
 u(x_1,x_2,x_3) = b^{\rho}\chi(b^{\rho} x_2)\chi(b^{\rho} x_3)\cdot\chi(b^{\rho}x_1)\cdot f_*\bigl((bB(p))^{1/2} x_1\bigr)\cdot e^{-i(bB(p)x_1)^{1/2}\hat\alpha x_3}.
 \]
 \end{proof}

\subsection{On the D-to-N operator relative to $\R^3_+$}  Let us denote
\[
H^1_A(\mathbb R^3_+)=\{u \in L^2(\mathbb R^3_+), D_{x_1}u \in L^2(\mathbb R^3_+),D_{x_2}u \in L^2(\mathbb R^3_+), (D_{x_3} + \cos \vartheta x_1-\sin \vartheta x_2)u  \in L^2(\mathbb R^3_+)\}\}\,.
\]
Since $H^1_A(\mathbb R^3_+)\subset H^1_{\rm loc}(\overline{\mathbb R^3_+})$, we can define the trace space
\[
\check H^{\frac 12}(\mathbb R^2)=\{u\in L^2_{loc}(\mathbb R^2)\,,\, \exists \tilde u \in H^1_A(\mathbb R^3_+)\mbox{ s.t. } \tilde u_{x_1=0}=u\}\,.
\]
To define the D-to-N operator, we can now start "formally" from the weak form given in \eqref{DtNweak} with $\Omega=\mathbb R^3_+$ and $A=A_{\vartheta}$:
\begin{equation}\label{DtNweaka}
	\left\langle \Lambda_{A} f , g \right \rangle_{\check H^{-1/2}(\partial \Omega) \times \check H^{1/2}(\partial \Omega)} = \int_\Omega  \langle (-i\nabla -A)u , (-i\nabla -A)v \rangle\ dx\,,
\end{equation}
for any $g \in \check H^{1/2}(\partial \Omega)$ and $f \in \check H^{1/2}(\partial \Omega)$ such that $u$ is the unique solution see (\ref{eq:1.4aa}) of
\[
H^{\rm Dir}(\vartheta) u= 0\,,\, u_{\partial \Omega}= f\,,
\]
and $v$ is any element of $H^1(\Omega)$ so that $v_{|\partial \Omega} = g$. \\

\begin{remark}\label{rem:3D}
It would be interesting  to verify various technical details in order  to associate a self-adjoint realization to this weak definition of the D-t-N operator in $3$D but this is not needed for the results presented in this paper.
\end{remark}

\section{ Weak magnetic field.}\label{sec:small-b}

\begin{proof}[Proof of Theorem~\ref{thm:small-b}]~\\

{\bf Step 1.}

Since $\Omega$ is simply connected,   $\lambda^{\rm DN}(bA,\Omega)=\lambda^{\rm DN}(bA_\Omega,\Omega)$.  We write $\hat\lambda=\lambda^{\rm DN}(bA_\Omega,\Omega)$.  Using $f=1$ as trial state in \eqref{eq:def-var-DN}, 
we get
\[ 0\leq \hat\lambda\leq \frac{b^2}{|\partial\Omega|}\int_\Omega|A_\Omega|^2dx. \]

{\bf Step 2.}

We introduce $E_1(\hat\lambda,b)=\mu(bA_\Omega,\Omega,\hat\lambda)$ as in \eqref{eq:int-R-ev} and we observe that $E_1(\hat\lambda,b)=0$ (see Proposition~\ref{prop:zero-R} with $\gamma=b^{-1/2}\hat\lambda$).   The second eigenvalue satisfies
\[\liminf_{b\to0^+}E_2(\hat\lambda,b)>0,\]
which follows from the min-max principle and the following lower bound
\[\||(-i\nabla-bA_\Omega)f\|_\Omega^2-\hat\lambda\|f\|_{\partial\Omega}^2\geq \Bigl(\frac12\|\nabla f\|^2_\Omega-2\hat\lambda\|f\|_{\partial\Omega}^2\Bigr)-2b^2\|A_\Omega f\|_\Omega^2.\]

{\bf Step 3.}

Let $f_{\hat\lambda}$ be the positive and normalized (in $L^2(\Omega)$) ground state of $E_1(\hat\lambda,0)$.  By \cite[Theorem~2.1]{GS}, 
\[E_1(\hat\lambda,0)=-\frac{|\partial\Omega|}{|\Omega|}\hat\lambda +o(\hat\lambda)\quad (\hat\lambda\to0).\] 
Writing
\[E_1(\hat\lambda,0)=\|\nabla f_{\hat\lambda}\|_\Omega^2-\hat\lambda\|f_{\hat\lambda}\|_{\partial\Omega}^2\geq\frac12\|\nabla f_{\hat\lambda}\|_\Omega^2+\frac12E_1(2\hat\lambda,0), \]
we deduce that
\[ \frac12\|\nabla f_{\hat\lambda}\|_\Omega^2\leq E_1(\hat\lambda,0)-\frac12E_1(2\hat\lambda,0)=o(\hat\lambda), \]
and by the Poincar\'e inequality
\[ \|f_{\hat\lambda}-\langle f_{\hat\lambda}\rangle \|_\Omega^2=o(\hat\lambda),\]
where $\langle f_{\hat\lambda}\rangle$ is the average of $f_{\hat\lambda}$ over $\Omega$.  The normalization of $f_{\hat\lambda}$ in $L^2(\Omega)$ then yields
\[ \langle f_{\hat\lambda}\rangle=\frac1{|\Omega|^{1/2}}+o(\hat\lambda).\]

{\bf Step 4.}

 Since $\mathrm{div}\,A_\Omega=0$,  it follows from \eqref{defMagOp}  that
\[
H_{bA_\Omega}f_{\hat\lambda}=E_1(\hat\lambda,0)f_{\hat\lambda}+b^2|A_\Omega|^2f_{\hat\lambda}-2i b A_\Omega\cdot \nabla f_{\hat\lambda}\,.\]
Moreover,  since $\vec\nu\cdot A_\Omega=0$, $f_{\hat\lambda}$ satisfies the boundary condition
\[ 
\partial_\nu f_{\hat\lambda}+i\vec{\nu}\cdot A_\Omega f_{\hat\lambda}=\hat\lambda f_{\hat\lambda}\mbox{ on }\partial\Omega\,.\] 
By Step 3,  we have
\[
\|2i b A_\Omega\cdot \nabla f_{\hat\lambda}\|_\Omega\leq b\|\nabla f_{\hat\lambda}\|_\Omega=o(b\hat\lambda^{1/2}).
\]
We introduce the real-valued function $g_{\hat\lambda}$ as the solution of 
\[ -\Delta g_{\hat\lambda}-C(\hat\lambda)f_{\hat\lambda}+|A_\Omega|^2f_{\hat\lambda}=0\mbox{ on }\Omega,\quad \partial_\nu g_{\hat\lambda}=\hat\lambda g_{\hat\lambda}\mbox{ on }\partial\Omega, \]
where 
\[C(\hat\lambda)=\int_{\Omega}|A_\Omega|^2 f_{\hat\lambda}^2dx=\frac1{|\Omega|}\int_{\Omega}|A_\Omega|^2dx+o(\hat\lambda).\]
With $u=f_{\hat\lambda}+b^2 g_{\hat\lambda}$,  we have
\[H_{bA_\Omega}u=\bigl(E_1(\hat\lambda,0)+b^2 C(\hat\lambda)\bigr)f_{\hat\lambda}-2i b A_\Omega\cdot \nabla f_{\hat\lambda}+b^2\bigl( b^2|A_\Omega|^2-2iA_\Omega\cdot\nabla \bigr)g_{\hat\lambda}\,. \]
By Step 1,  this yields
\[ \left\|H_{bA_\Omega}u-\Bigl(E_1(\hat\lambda,0)+ \frac{b^2}{|\Omega|}\|A_\Omega\|_\Omega^2\Bigr) u\right\|_\Omega =o(b^2)\quad (b\to0^+),\]
and by Step 2,  the spectral theorem yields 
\[E_1(\hat\lambda,b)=E_1(\hat\lambda,0)+\frac{b^2}{|\Omega|}\|A_\Omega\|_\Omega^2+o(b^2).\]

\newpage
{\bf Step 5.}

To finish the proof,  we use that $E_1(\hat\lambda,b)=0$ and that $\hat\lambda=O(b^2)$,  along with the asymptotics of $E_1(\hat\lambda,0)$ in Step~2.  
\end{proof}

\subsection*{Acknowledgments} 
A.K.  is partially supported by CUHK-SZ grant no. UDF01003322 and UF02003322.   F.N is supported  by the French National Research Project GDR Dynqua. The authors would like to thank Vladimir Lotoreichik for the helpful comments.

\appendix
\section{On the choice of gauge}\label{appendix.A}
Considering $A$ as a 1-form is helpful in passing from the Cartesian to parallel coordinates,
\[A_1dx_1+A_2dx_2=\tilde A_1dt+\tilde A_2 ds, \]
where 
\[\partial_t\tilde A_2-\partial_s\tilde A_1=-(1-tk(s))B, \]
and $B=\curl A$ is the magnetic field in Cartesian coordinates. 

Suppose that $B$ is constant.  Choosing a simply connected set $V\subset\mathbb R/L\mathbb Z$, we define a function $\varphi: \mathbb R/L\mathbb Z\times V$ such that
\[\partial_t\varphi=\tilde A_1,\quad \partial_s\varphi-B(t-\mbox{$\frac12$}t^2k(s))=\tilde A_2(s,t)\quad \mbox{ on }V.\]
Such a function is given by
\[\varphi(t,s)=\int_0^t\tilde A_1(\tau,s)d\tau+g(s), \]
where $g:\mathbb R/L\mathbb Z\to \mathbb R$ satisfies $g'(s)=\tilde A_2(0,s)$ on $V$, and  $g(s)=0$ outside a neighborhood of $V$; this last condition ensures the periodicity of $g$.

That this is relevant is apparent from
\begin{multline*}\|(-i\nabla-bA)u\|_{\Phi_0^{-1}(V)}^2\\
=\int_{V}\bigl(|(-i\partial_t-b\tilde A_1)\tilde u|^2+(1-tk(s))^{-2}|(-i\partial_s-b\tilde A_2)\tilde u|^2\bigr)(1-tk(s))dtds\\
=\int_{V}\bigl(|\partial_t(e^{-ib\varphi}\tilde u)|^2\\
+(1-tk(s))^{-2}|(-i\partial_s+bB(t-\mbox{$\frac12$}t^2k(s)))e^{-2ib\varphi}\tilde u|^2\bigr)(1-tk(s))dtds.\end{multline*}

\section{A reminder for models in half-spaces.}\label{Subsection3.2}
We refer to \cite{LuPa} and \cite{HeMo3} for the proof of the results which are recalled here 
as  presented in \cite{HeMo1}.
If $N$ is a unit vector in $\mathbb R^3$, we now  consider 
the Neumann realization in  $\Omega:=\{ x\in \mathbb R^3\;|\; x\cdot N >0\}$. After a 
rotation, we can 
assume in the proofs
 that $N=(1,0,0)$, so $\Omega$ is  $\mathbb R^3_+:=\{ x_1 >0\}$.\\
After scaling, we can assume that $h=1$ and 
$| H |  =1$. Here $\vec H$ is the magnetic vector field $\curl \vec{A}$ associated with the magnetic field $B$ considered as a $2$-form.

 After some rotation in the $(x_2,x_3)$ variables, we can assume that
the new magnetic field is $(\beta_1,\beta_2, 0)$
and  we are   reduced to the problem of analyzing~:
 $$
P(\beta_1, \beta_2):=D_{x_1}^2 + D_{x_2}^2 + (D_{x_3}  + \beta_2  x_1 - \beta_1 
x_2)^2\;,
$$
in $\{ x_1 >0\}$, where~:
$$ \beta_1^2 + \beta_2^2 =1\;.$$
We introduce~:
$$
\beta_2 = \cos \vartheta\;,\; \beta_1 = \sin \vartheta\;,
$$
and we observe that, if $N$ is the external normal to $x_1=0$,
 we have~:
\begin{equation}\label{orth}
\langle \vec{H}\,|\,  N \rangle= - \sin \vartheta\;.
\end{equation}
By partial Fourier transform, we arrive to~:
\begin{equation}
L(\vartheta, \tau) = D_{x_1}^2 + D_{x_2}^2 + (\tau + \cos \vartheta \;  x_1
 -  \sin  \vartheta \; x_2 )^2\;,
\end{equation}
in $x_1 >0$ and with Neumann condition on $x_1 =0$.
The bottom of the spectrum of $L(\vartheta, \tau)$ is given by~:
\begin{equation}\label{bott}
\sigma(\vartheta):=\inf \spe \left( L(\vartheta, D_t)\right) = \inf_{\tau} 
\left(\inf \spe( 
L(\vartheta,\tau) )\right)\,.
\end{equation}

\begin{proposition}\label{Propvartheta}.\\
The bottom of the spectrum of the Neumann realization of $H_{A,\Omega}^{\rm Ne}$ in
 \break $\Omega:=\{x\in \mathbb R^3\;|\; x\cdot N >  0\}$
 is~:
\begin{equation}
\inf \; \spe H_{A,\Omega}^{\rm Ne}= \sigma (\vartheta) \; b  \;  h\;,
\end{equation}
where $\vartheta \in [-\frac \pi 2, \frac \pi 2]$ is defined by (\ref{orth}).
\end{proposition}
By symmetry considerations, we observe also that~:
\begin{equation}\label{sym}
\sigma(\vartheta) = \sigma (- \vartheta) = \sigma ( \pi - \vartheta)\;.
\end{equation}
It is consequently enough to look at the restriction to $[0,\frac \pi 2]$.
\subsection{Properties of $\vartheta \mapsto 
\sigma(\vartheta)$.}\label{Subsection3.3}
Let us now list the main properties of the 
 function $\vartheta \mapsto \sigma(\vartheta)$ 
 on $[0,\frac \pi 2 ]\;$ . Most of them are established in \cite{LuPa}
 but see also \cite{HeMo3}.

\begin{enumerate}
\item $\sigma$ is continuous on $[0, \frac \pi 2]$.
\item 
\begin{equation}
\sigma(0) = \Theta_0 < 1\;.
\end{equation}
\item
 \begin{equation}
\sigma(\frac \pi 2 )=1\;.
\end{equation}
\item 
\begin{equation}\label{roughlow}
\sigma(\vartheta) \geq \Theta_0 (\cos \vartheta)^2 + (\sin \vartheta)^2\;.
\end{equation}
\item
If $\vartheta \in ]0,\frac \pi 2[$, the spectrum of $L(\vartheta,\tau)$
 is independent of $\tau$ and its   
 essential spectrum  is contained in $[1, +\infty[$.
\item 
For $\vartheta \in ]0,\frac \pi 2 [$, $\sigma(\vartheta)$
 is an isolated eigenvalue of $L(\vartheta,\tau)$, with multiplicity one.
\item 
The function $\sigma $  is strictly increasing on  
$[0, \frac \pi 2[$.
\end{enumerate}
An immediate  consequence of this analysis is
\begin{proposition} \label{Proposition3.5}.\\
When $b=| H |$ is fixed the bottom of the spectrum  of $H_{A,\Omega}^{\rm Ne}$ 
in $\Omega 
:=\{ x\cdot N >0 \}$ is minimal
 when $\vartheta=0$ that is, according to (\ref{orth}),   when 
 the magnetic field vector satisfies $\vec{H}\cdot N =0$.
\end{proposition}

\end{document}